\documentclass[a4paper,USenglish,cleveref,autoref,pdfa]{lipics-v2021}

\pdfoutput=1
\hideLIPIcs

\graphicspath{{./figures/}}
\title{Tilt Automata:\texorpdfstring{\\}{} Gathering Particles With Uniform External Control}

\titlerunning{Tilt Automata: Gathering Particles With Uniform External Control}

\author{Sándor P. Fekete}{Department of Computer Science, TU Braunschweig and L3S, Germany}{s.fekete@tu-bs.de}{https://orcid.org/0000-0002-9062-4241}{}
\author{Jonas Friemel}{Department of Electrical Engineering and Computer Science,\\ Bochum University of Applied Sciences, Germany}{jonas.friemel@hs-bochum.de}{https://orcid.org/0009-0009-6270-4779}{}
\author{Peter Kramer}{Department of Computer Science, TU Braunschweig, Germany}{kramer@ibr.cs.tu-bs.de}{https://orcid.org/0000-0001-9635-5890}{}
\author{Jan-Marc Reinhardt}{Department of Electrical Engineering and Computer Science,\\ Bochum University of Applied Sciences, Germany}{jan-marc.reinhardt@hs-bochum.de}{https://orcid.org/0009-0005-8907-3832}{}
\author{Christian Rieck}{Institute of Mathematics, University of Kassel, Germany}{christian.rieck@mathematik.uni-kassel.de}{https://orcid.org/0000-0003-0846-5163}{}
\author{Christian Scheffer}{Department of Electrical Engineering and Computer Science,\\ Bochum University of Applied Sciences, Germany}{christian.scheffer@hs-bochum.de}{https://orcid.org/0000-0002-3471-2706}{}

\authorrunning{S.\,P. Fekete, J. Friemel, P. Kramer, J.-M. Reinhardt, C. Rieck, and C. Scheffer}

\Copyright{Sándor P. Fekete, Jonas Friemel, Peter Kramer, Jan-Marc Reinhardt, Christian Rieck, and Christian~Scheffer}



\funding{%
  \emph{Sándor P.\ Fekete, Jonas Friemel, Peter Kramer, Christian Scheffer}: Supported by the Deutsche Forschungsgemeinschaft (DFG, German Research Foundation) -- 530918134.\\
  \emph{Christian Rieck}: Supported by the Deutsche Forschungsgemeinschaft (DFG, German Research Foundation) -- 522790373.
}

\nolinenumbers

\usepackage[algo2e,linesnumbered,ruled,vlined]{algorithm2e}
\usepackage[basic]{complexity}
\usepackage[labelformat=simple]{subcaption}

\DeclareCaptionLabelFormat{subfiglabel}{\textsf{\textbf{#2}}}
\usepackage{tikz}
\usetikzlibrary{arrows.meta,automata,positioning}
\usepackage{xfrac}
\usepackage{xspace}

\theoremstyle{definition}
\newtheorem{problemInternal}[theorem]{Problem}
\newcommand{\problem}[3]{
\begin{problemInternal}
  {#1}.
  \begin{description}
  \item[Input:] {#2}
  \item[Goal:] {#3}
  \end{description}
\end{problemInternal}}
\newcommand{\paraProblem}[4]{
\begin{problemInternal}
  {#1}.
  \begin{description}
  \item[Input:] {#2}
  \item[Question:] {#3}
  \item[Parameters:] {#4}
  \end{description}
\end{problemInternal}}

\newcommand{\bigO}{\mathcal{O}}
\newcommand{\Tpose}{\ensuremath{\mathrm{T}}}
\newcommand{\ZZ}{\ensuremath{\mathbb{Z}^2}}
\newcommand{\N}{\ensuremath{\mathbb{N}}}
\newcommand{\D}{\mathbb{D}}
\newcommand{\udir}{\ensuremath{\text{\textup{\textsc{u}}}}}
\newcommand{\ddir}{\ensuremath{\text{\textup{\textsc{d}}}}}
\newcommand{\ldir}{\ensuremath{\text{\textup{\textsc{l}}}}}
\newcommand{\rdir}{\ensuremath{\text{\textup{\textsc{r}}}}}
\newcommand{\FT}{\ensuremath{\mathrm{FT}}\xspace}
\newcommand{\SSt}{\ensuremath{\mathrm{S1}}\xspace}
\newcommand{\sgs}[1][P]{\ensuremath{\operatorname{sgs}(#1)}}
\newcommand{\rt}[1][A]{\ensuremath{\operatorname{rt}(#1)}}
\newcommand{\Cerny}{{\v{C}}ern{\'{y}}}
\newcommand{\FPT}{\ComplexityFont{FPT}\xspace}
\newcommand{\XP}{\ComplexityFont{XP}\xspace}
\newcommand{\thmW}[1]{\W\textup{[#1]}}
\newcommand{\fullGathering}{\textup{\textsc{FullGathering}}\xspace}
\newcommand{\quickGathering}{\textup{\textsc{ShortestGatheringSequence}}\xspace}
\newcommand{\subsetGathering}{\textup{\textsc{SubsetGathering}}\xspace}
\newcommand{\paraGathering}{\textup{\textsc{ParaGathering}}\xspace}
\newcommand{\occupancy}{\textup{\textsc{Occupancy}}\xspace}
\newcommand{\shapeReconfiguration}{\textup{\textsc{ShapeReconfiguration}}\xspace}
\newcommand{\tiltCover}{\textup{\textsc{TiltCover}}\xspace}

\begin{document}
\maketitle

\begin{abstract}
  Motivated by targeted drug delivery, we investigate the gathering of particles in the full tilt model of externally controlled motion planning:
A set of particles is located at the tiles of a polyomino with all particles reacting uniformly to an external force by moving as far as possible in one of the four axis-parallel directions until they hit the boundary.
The goal is to choose a sequence of directions that moves all particles to a common position.
Our results include a polynomial-time algorithm for gathering in a completely filled polyomino as well as hardness reductions for approximating shortest gathering sequences and for determining whether the particles in a partially filled polyomino can be gathered.
We pay special attention to the impact of restricted geometry, particularly polyominoes without holes.
As corollaries, we make progress on an open question from~\cite{tilt-soda20} by showing that deciding whether a given position can be occupied remains \NP-hard in polyominoes without holes and provide initial results on the parameterized complexity of tilt problems.
Our results build on a connection we establish between tilt models and the theory of synchronizing~automata.

\end{abstract}

\section{Introduction}\label{sec:intro}

We investigate a problem motivated by targeted drug delivery~\cite{gathering-rl-iros22}: How can initially dispersed drug-carrying particles be gathered at a common position?
Problems of this type are closely related to advances in micro-scale robotics~\cite{science-microrobots-2025}.
However, environmental constraints and the severely limited capabilities of robotic agents at micro- and nano-scale rule out autonomous movement.
Instead, a global external force, e.g., an electromagnetic field, can be used to move and steer simple robotic particles.
The challenge is to use geometry to control a large swarm of particles through uniform signals within a non-uniform environment~\cite{tilt-video-socg15,tilt-iros13}.

Abstract mathematical models~\cite{tilt-algosensors13,tilt-limited-directions-jip20}, commonly called \emph{tilt models} in reference to gravity as a potential external force, have been developed to tackle this problem from an algorithmic perspective.
They model the environment as a \emph{polyomino}, i.e., a finite, connected region of the square tiling, with particles positioned at tiles and moving on axis-parallel paths up, down, left, or right.
Two major variants exist: In the \emph{single step model} (\SSt), particles can be precisely controlled via unit steps to adjacent positions, whereas the \emph{full tilt model} (\FT) takes imprecision into account by only allowing particles to move maximally until blocked.

Algorithms and complexity results for a variety of applications~\cite{assembly-algorithmica20,tilt-reconfiguration-nc21,tilt-sort-classify,tilt-mapping} have been discovered using tilt models.
To investigate algorithmic approaches to targeted drug delivery, we consider particles located at the tiles of a polyomino and look for a \emph{gathering sequence}, i.e., a sequence of moves that relocates all particles to the same tile.
It is already known that gathering sequences of length $\bigO(n_c D^2)$ always exist in \SSt, where $n_c$ is the number of convex corners of the polyomino and $D$ its diameter, but finding sequences of minimum length is \NP-hard~\cite{gathering-icra20}.
Very little is known about gathering in \FT, apart from the brief observation in~\cite{gathering-case16} that gathering sequences do not exist for all polyominoes in \FT.
We remedy this situation with a comprehensive study of gathering in \FT.

\begin{figure}[bth]
  \begin{subcaptionblock}{0.4\textwidth}
    \phantomcaption\label{sfig:gathering-ft-a}
    \centering
    \includegraphics[page=1]{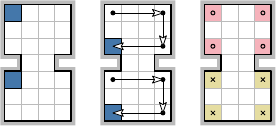}\\
    \captiontext*{}
  \end{subcaptionblock}\hfill
  \begin{subcaptionblock}{0.6\textwidth}
    \phantomcaption\label{sfig:gathering-ft-b}
    \centering
    \includegraphics[page=2]{gathering-full-tilt}\\
    \captiontext*{}
  \end{subcaptionblock}
  \caption{Two examples where gathering a set of particles is impossible in the full tilt model. Both examples depict, from left to right, a set of particles (blue squares) inside a polyomino, a sequence of moves leading to all possible configurations, and a partition of reachable positions yielding a congruence. Example~\subref{sfig:gathering-ft-a} prohibits gathering by keeping the particles in different congruence classes, whereas the particles in~\subref{sfig:gathering-ft-b} always occupy distinct positions of the same class.}\label{fig:gathering-ft}
\end{figure}

As a first observation, we can see that the full tilt model exhibits several interesting complications that render gathering impossible.
The authors of~\cite{gathering-case16} already recognized that particles may become trapped in separate parts of the polyomino, see \cref{sfig:gathering-ft-a}.
However, this is not the only difficulty.
Remarkably, even if all tiles reachable from the initial particle positions are mutually reachable, gathering may still be impossible, see \cref{sfig:gathering-ft-b}.

Apart from practical applications, tilt models fascinate with simple, intuitive problems that nevertheless exhibit surprising complexity.
A central question is how the computational complexity of a problem relates to the geometric complexity of the environment.
Two fundamental problems are \occupancy, which asks if any of a given set of particles can be moved to a specified tile, and \shapeReconfiguration, which asks if a given configuration of particles can be transformed into another given configuration.
Both are proven \PSPACE-complete in \FT for general polyominoes in~\cite{tilt-soda20}, where the authors pose the complexity for \emph{simple} polyominoes, i.e., polyominoes without holes, as an open problem.
We also consider a supplementary class of polyominoes we call \emph{mazes}: collections of horizontal and vertical lines that only intersect in \textsf{T}-shaped and \textbf{+}-shaped junctions, see \cref{sfig:polyomino-maze}.

\subsection{Our contributions}

We provide a wealth of insights into the gathering of particles in the full tilt model.
Let a given polyomino have $n$ corners, $n_c \leq n$ of which are convex.
We prove the following results.

\begin{itemize}
\item When particles are initially positioned at every tile of a polyomino, a gathering sequence of length $\bigO(n_c n^2)$ can be computed, if one exists at all (\cref{thm:alg-gathering,cor:gather-upper-bound}).
\item There are polyominoes with shortest gathering sequences of length $\Omega(n^2)$ (\cref{thm:gather-lower-bound}).
\item Finding a shortest gathering sequence is \NP-hard, even in simple mazes (\cref{thm:gathering-hard-simple-mazes}).
\item There is a polynomial-time $4$-approximation for minimizing the length of a gathering sequence in simple mazes (\cref{thm:simple-maze-approx}), but there is no polynomial-time constant-factor approximation in general simple polyominoes (unless $\P = \NP$; \cref{thm:inapproximability-simple}).
\item Determining whether a subset of tiles has a gathering sequence is \PSPACE-complete in mazes (\cref{cor:subset-gathering-mazes}) and \NP-hard in simple polyominoes (\cref{thm:subset-gathering-simple}).
\end{itemize}

The last of these results allows us to derive progress on the open question from~\cite{tilt-soda20} regarding the complexity of \occupancy and \shapeReconfiguration in simple polyominoes.
We show that both remain \NP-hard (\cref{cor:occupancy-hard-simple}).
Furthermore, we contribute to questions raised in~\cite{tilt-deterministic} by proving a natural generalization of \occupancy and \shapeReconfiguration, which we call \tiltCover, \NP-hard in simple polyominoes even for \emph{deterministic} move sequences, i.e., for repetitions of a fixed sequence of moves.
Finally, we establish initial results on the parameterized complexity of tilt problems by showing that a parameterized gathering problem, as well as \occupancy and \shapeReconfiguration{}, are \thmW{1}-hard when parameterized by the number of particles in a given configuration.
Throughout this paper, we utilize a so far overlooked connection between tilt models and automata theory, in particular synchronizing automata.
Since this topic is unlikely to be widely known among the computational geometry community, we give a brief survey in \cref{subsec:sync-automata}.
We expect that more insights for tilt models can be gained by exploring this approach further.
Conversely, the geometric nature of tilt models could lead to new ways to tackle open questions in automata theory.
As we shall see, the discrepancy between the upper and lower bounds on the worst-case lengths of gathering sequences is closely related to a longstanding open problem known as \Cerny's conjecture.

The rest of this text is structured as follows.
After a discussion of related work in \cref{subsec:relatedwork}, \cref{sec:preliminaries} introduces formal definitions and notation.
\Cref{sec:gathering} establishes the connection between tilt models and automata theory and uses it to gather particles in a completely filled polyomino.
Conversely, \cref{sec:simulation} demonstrates how particles can be used to simulate automata and discusses the consequences for gathering tasks.
We then focus on restricted environments, specifically simple polyominoes, for approximating shortest gathering sequences (\cref{sec:optimization}) and for gathering subsets (\cref{sec:subsets}).
\Cref{sec:consequences} examines a variety of consequences our techniques have for related problems.
Finally, \cref{sec:conclusion} provides concluding remarks and open questions.

\subsection{Related work}\label{subsec:relatedwork}

Apart from the aforementioned work~\cite{gathering-icra20,gathering-case16} on gathering in the single step model, tilt models have been investigated for a variety of purposes.
These include realizing permutations of labeled rectangular configurations~\cite{tilt-algosensors13,tilt-nc19,tilt-designing-worlds,tilt-rearranging}, simulating dual-rail logic circuits~\cite{tilt-nc19,tilt-designing-worlds,tilt-fan-out}, sorting and classifying polyominoes~\cite{tilt-sort-classify}, mapping an unknown region~\cite{tilt-mapping}, and filling and draining a polyomino~\cite{tilt-fill}.
Notably, gathering all particles in a completely filled polyomino is strictly more difficult than draining to a freely chosen sink: While gathering requires a single position $p$ to be reachable from everywhere, which suffices to guarantee drainability to a sink placed at $p$~\cite{tilt-fill}, this does not ensure gatherability, see \cref{sfig:gathering-ft-b}.
A major line of research is concerned with the assembly of structures from particles that can be glued together.
Introduced in~\cite{assembly-ral17}, a formal algorithmic analysis~\cite{assembly-algorithmica20} was followed by the introduction of more efficient assembly in stages~\cite{assembly-parallel}, universal constructors that can build any of a class of shapes encoded in move sequences~\cite{tilt-soda20,tilt-soda19}, and assembly via more fine-grained control in the single step model~\cite{assembly-iros23,assembly-single-step}.
Assembly is strongly related to universal reconfiguration, for which algorithms achieving a worst-case optimal number of moves in either full tilt or single step model were given in~\cite{tilt-reconfiguration-nc21}.

The study of the computational complexity of fundamental tilt problems was initiated in~\cite{tilt-algosensors13}, where the problems now known as \occupancy and \shapeReconfiguration were originally shown \NP-hard in \FT.
It was later shown that finding a shortest sequence of moves for labeled reconfiguration is \PSPACE-hard~\cite{tilt-designing-worlds}, and that \shapeReconfiguration is \PSPACE-complete with the addition of $2 \times 2$ particles~\cite{tilt-soda19}, before the complexity of \occupancy and \shapeReconfiguration was settled to be \PSPACE-complete with only unit particles~\cite{tilt-soda20}.
In the single step model, \occupancy is solvable in polynomial time~\cite{tilt-limited-directions-jip20}, unless dominoes are allowed as particles, in which case it becomes \PSPACE-complete even in rectangles~\cite{tilt-hardness-cccg20}, whereas \shapeReconfiguration remains \PSPACE-complete in \SSt with unit particles~\cite{tilt-hardness-cccg20}.
Compared to \FT, more complexity results are known for \SSt when the environment is restricted.
The \textsc{Relocation} problem, which asks if a specific labeled particle can be moved to a given position, is \NP-complete in monotone polyominoes when restricted to three directions~\cite{tilt-limited-directions-jip20}, and \NP-complete in squares when restricted to two directions~\cite{tilt-limited-directions-nc25}.
Most recently, the authors of~\cite{tilt-deterministic} proved \PSPACE-completeness of several tilt problems for \emph{deterministic} move sequences, i.e., for repetitions of a fixed sequence of moves.

Taking the maximal movement of the full tilt model and adding individual control leads to the rules of the game Ricochet Robots~\cite{randolph-endm06,randolph-jip17}.
Holzer and Schwoon~\cite{atomix} used ideas similar to ours in their \PSPACE-completeness proof for the related game Atomix, but we cannot use their construction directly because we need one that allows to analyze the length of gathering sequences.
Other related models examine uniform movement in a rectangle, e.g., by sweeping lines~\cite{sweeping-eurocg16,sweeping-cgt23}, or the rules of the game~2048~\cite{2048-fun16,2048-cccg20}.

The gathering of point particles inside a polygon has been considered by Bose and Shermer~\cite{repulsion-cg20}, but their setting uses repulsion from a point within the polygon, whereas tilt models could be interpreted to use repulsion from a point at infinity.
More closely related is the problem of localizing robots with severely limited capabilities~\cite{localization-icra20,localization-tor07}.
A sequence of movement commands that leads a robot from an initially unknown position in a polygonal environment to a definite position is equivalent to gathering all point particles in the environment at the same final position.

\section{Preliminaries}\label{sec:preliminaries}

In this section, we give the rigorous definitions omitted from \cref{sec:intro} in favor of intuition.
Every polyomino~$P$ has an associated \emph{dual graph}~$G_P=(V,E)$, see \cref{sfig:polyomino-maze}, where $V \subset \ZZ$ and $E = \{\{u,v\} : {u \in V}, {v \in V}, \|u-v\|_1 = 1\}$, and a \emph{boundary}~$\partial P$, an axis-aligned polygon with integer side lengths that separates $V$ from $\ZZ \setminus V$.
We identify the vertices of $V$ with the tiles of $P$ and call them \emph{pixels}.
As a general convention throughout this text, $n$ will denote the number of corners of $\partial P$, $n_c \leq n$ the number of convex corners, and $N = |V|$ the number of pixels of a given polyomino.
Pixels are uniquely determined as the intersection of two \emph{segments}, a \emph{row segment} and a \emph{column segment}, which are maximally contiguous parts of rows and columns; a \emph{corner pixel} is adjacent to two perpendicular sides of $\partial P$, i.e., next to a convex corner, see \cref{sfig:polyomino-simple}.
Formally, a maze is a \emph{thin} polyomino (it contains no $2 \times 2$ squares) in which all corner pixels have degree~$1$, see \cref{sfig:polyomino-maze}.

\begin{figure}[tbh]
  \begin{subcaptionblock}{0.16\textwidth}
    \phantomcaption\label{sfig:polyomino-simple}
    \centering
    \includegraphics[page=1]{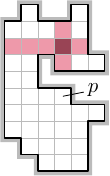}\\
    \captiontext*{}
  \end{subcaptionblock}\hfill
  \begin{subcaptionblock}{0.16\textwidth}
    \phantomcaption\label{sfig:polyomino-maze}
    \centering
    \includegraphics[page=2]{models}\\
    \captiontext*{}
  \end{subcaptionblock}\hfill
  \begin{subcaptionblock}{0.68\textwidth}
    \phantomcaption\label{sfig:moves}
    \centering
    \includegraphics[page=3]{models}\\
    \captiontext*{}
  \end{subcaptionblock}\hfill
  \caption{\subref{sfig:polyomino-simple}~A simple polyomino with a darkly shaded pixel at the intersection of two lightly shaded segments and a corner pixel~$p$. \subref{sfig:polyomino-maze}~A non-simple maze and its dual graph. \subref{sfig:moves}~Downwards moves in the blocking and merging variants of the full tilt model acting on a configuration.}\label{fig:models}
\end{figure}

A set $C \subseteq V$ of all pixels containing particles is called a \emph{configuration}.
Pixels $p \in C$ are \emph{occupied}, those in $V \setminus C$ are \emph{free}.
We use the set $\D = \{\udir,\ddir,\ldir,\rdir\}$ as shorthand for the directions up, down, left, and right.
Formally, a tilt model defines a function $\delta: \D \to (2^V \to 2^V)$, i.e., it maps directions to transformations of configurations.
A transformation $\delta(v)$, $v \in \D$, is called a \emph{move}.
$\delta$ is inductively extended to a function on $\D^*$, the set of (possibly empty) sequences of directions, by setting $\delta(wv) = \delta(v) \circ \delta(w)$, for $v \in \D$ and $w \in \D^*$, and $\delta(\varepsilon) = \operatorname{id}$, where $\varepsilon$ is the empty sequence and $\operatorname{id}$ the identity function.
Particles in the full tilt model are usually blocked when colliding with other particles, but for the purpose of gathering, they merge instead, see \cref{sfig:moves}.
A formal definition of the \emph{blocking} and \emph{merging} variants for the direction~$\ldir$ is as follows; definitions for the other directions are analogous.

\begin{description}
  \item[Full tilt (blocking):] $p \in \delta^\times(\ldir)(C)$ if and only if $p$ is one of the $|R \cap C|$ leftmost pixels of the row segment~$R$ containing $p$.
  \item[Full tilt (merging):] $p \in \delta^\cup(\ldir)(C)$ if and only if $p$ is the leftmost pixel of the row segment~$R$ containing $p$ and $R \cap C \neq \varnothing$.
\end{description}

Both variants agree on singleton configurations, and the set of singleton configurations is closed under all moves.
This gives rise to the function $\delta^1: \D \to (V \to V)$ defined as $\delta^1(v)(p) = q$ whenever $\delta^\times(v)(\{p\}) = \delta^\cup(v)(\{p\}) = \{q\}$ for $v \in \D$.
When $\delta$ is clear from the context, we simplify the notation by writing $C \cdot w$ instead of $\delta(w)(C)$ and $C \xrightarrow{w} C'$ to mean $C' = \delta(w)(C)$, and use moves~$\delta(v)$ and directions~$v \in \D$ interchangeably.

\begin{figure}[htb]
  \centering
  \includegraphics{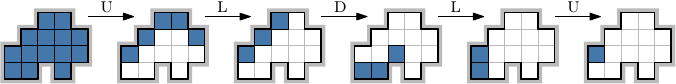}
  \caption{The sequence $\udir\ldir\ddir\ldir\udir$ is a gathering sequence for this polyomino in the full tilt model.}\label{fig:gathering-seq}
\end{figure}

A \emph{gathering sequence for a configuration $C$} is a sequence of directions $w \in \D^*$ such that $|\delta^\cup(w)(C)| = 1$, in which case $C$ is a \emph{gatherable configuration}; a \emph{gathering sequence for a polyomino $P$} is a gathering sequence for the configuration $V$, see \cref{fig:gathering-seq}, making $P$ a \emph{gatherable polyomino}.
We define $\sgs[C]$ to be the minimum length of a gathering sequence (i.e., the length of a \emph{shortest gathering sequence}) for a gatherable configuration~$C$, and assign $\sgs = \sgs[V]$ for a gatherable polyomino~$P$.

The uniform movement of particles allows us to reduce the behavior of configurations in the merging variant to the motion of single particles.

\begin{observation}\label{obs:gathering-localization}
For $w \in \D^*\!$ and $C \subseteq V$, $\delta^\cup(w)(C) = \bigcup_{p \in C} \{\delta^1(w)(p)\} = \bigcup_{p \in C} \delta^\times(w)(\{p\})$.
\end{observation}

The single step analogue of \cref{obs:gathering-localization} is implicitly contained in~\cite{gathering-icra20}.
It lies at the heart of the equivalence between gathering and localization mentioned in \cref{subsec:relatedwork}: One can either consider a configuration of particles to be gathered by merging, or move a single particle initially positioned at one of several possible positions.
Several simple consequences follow:
A gathering sequence for a configuration~$C$ is also a gathering sequence for every $C' \subseteq C$.
If $w$ is a gathering sequence for $C$, then so is $ww'$ for every $w' \in \D^*$.
Prepending works for $V$: If $w$ is a gathering sequence for $V$, then so is $w'w$ for every $w' \in \D^*$.

We will use two algebraic properties of the full tilt model. First, every move is idempotent, i.e., $\delta(v) \circ \delta(v) = \delta(v)$ for all $v \in \D$.
Secondly, when two moves in opposite directions are applied in sequence, the first is nullified, e.g., $\delta(\ldir) \circ \delta(\rdir) = \delta(\ldir)$.
Thus, we can generally assume that sequences of moves alternate between perpendicular directions.

We investigate three problems that have so far not been analyzed in the full tilt model.

\problem{\fullGathering}{A polyomino~$P$.}{Find a gathering sequence for $P$, or decide that $P$ is not gatherable.}

\problem{\quickGathering}{A gatherable polyomino~$P$.}{Compute a gathering sequence for $P$ of minimum length $\sgs$.}

\problem{\subsetGathering}{A polyomino~$P$ and a configuration~$C \subseteq V$.}{Decide if $C$ is gatherable.}

To investigate if the difficulty of gathering tasks arises only under asymptotic conditions with unbounded numbers of particles and long move sequences, we consider a general parameterized gathering problem.
We assume familiarity with fundamental notions of parameterized complexity, such as fpt-reductions and the complexity classes \FPT, \XP, and \W[$i$], $i \geq 1$; see, e.g., the textbook by Flum and Grohe~\cite{textbook-parameterized}.

\paraProblem{\paraGathering}{A polyomino $P$ and a configuration $C$ with $|C|=k$.}{Does a gathering sequence $w$ exist for $C$ with $|w| \leq \ell$?}{$k \geq 1$, $\ell \in \N \cup \{\infty\}$.}

When $\ell$ gets fixed to an integer, \paraGathering is trivially solvable in $2^{\ell+1} N^{\bigO(1)}$ time by enumerating and checking all $4 \cdot 2^{\ell-1}$ possible sequences of length $\ell$ with alternating horizontal and vertical moves.
Thus, \paraGathering is in \FPT when parameterized by $\ell$.

\section{Gathering particles by synchronizing automata}\label{sec:gathering}

The notation introduced in \cref{sec:preliminaries} differs slightly from established conventions and may be considered unnecessarily heavy. It does, however, expedite a point we want to make: Tilt models are transition functions of automata.
A \emph{semi-automaton}\footnote{All (semi-)automata considered in this text are complete, deterministic, and finite.} is a triple $(Q, \Sigma, \delta)$, where $Q$ is a finite, non-empty set of \emph{states}, $\Sigma$ is a finite, non-empty \emph{alphabet}, and $\delta: \Sigma \to (Q \to Q)$ is a \emph{transition function} mapping letters of $\Sigma$ to transformations of~$Q$. As usual, $\delta$ gets extended to a function over \emph{words}~$w \in \Sigma^*$.
An \emph{acceptor} is a 5-tuple $(Q,\Sigma,\delta,q_0,F)$, where $(Q,\Sigma,\delta)$ is a semi-automaton, $q_0 \in Q$ an \emph{initial state}, and $F \subseteq Q$ a set of \emph{accepting states}.
The \emph{language accepted by an acceptor $A$} is the set of words $L(A)=\{w \in \Sigma^* : q_0 \cdot w \in F\}$.
For simplicity, we will refer to both semi-automata and acceptors as \emph{automata}.
Now it is easy to see that a tilt model defines an automaton $(2^V, \D, \delta)$ for a polyomino~$P$ with dual graph~$G_P=(V,E)$. We can even restate some established problems in the language of automata theory.

\begin{example}\label{ex:occupancy-automaton}
Assume we are working in a specified tilt model and are given a polyomino $P$.
\begin{itemize}
  \item The \shapeReconfiguration problem asks if $L((2^V, \D, \delta, C, \{C'\}))$ is non-empty for given configurations $C,C' \subseteq V$.
  \item The \occupancy problem asks if $L((2^V, \D, \delta, C, \{C' \subseteq V: p \in C'\}))$ is non-empty for a given configuration $C \subseteq V$ and pixel $p \in V$.
\end{itemize}
\end{example}

Although interesting, this connection is of no immediate use due to the exponential size of the set of states. Indeed, we would not expect a polynomial description because \occupancy and \shapeReconfiguration are \PSPACE-complete in \FT~\cite{tilt-soda20}, whereas emptiness of $L(A)$ can be checked in time polynomial in the size of $A$~\cite{complexity-survey}.
The approach via automata shows its merit when we consider the analogue to gathering sequences, synchronizing words.

\subsection{A brief overview of synchronizing automata}\label{subsec:sync-automata}

A word~$w \in \Sigma^*$ is a \emph{synchronizing word} for a set~$Q' \subseteq Q$ if there is a state~$q \in Q$ such that $p \xrightarrow{w} q$ for all $p \in Q'$, in which case $Q'$ is a \emph{synchronizing set}.
The minimum length of a synchronizing word for a synchronizing set~$Q'$ is called the \emph{reset threshold} $\rt[Q']$.
A~\emph{synchronizing automaton} $A=(Q,\Sigma,\delta)$ has a synchronizing set of states~$Q$, in which case we define $\rt = \rt[Q]$.
We now briefly summarize some important results; more information can be found in several excellent surveys~\cite{complexity-survey,sandberg-survey04,volkov-survey08,volkov-survey22}.

It has been (re-)discovered multiple times how to check if an automaton is synchronizing and, if so, find a synchronizing word of length $\bigO(|Q|^3)$.
Eppstein's work~\cite{eppstein-reset} is particularly noteworthy for its careful analysis of the time and space requirements.
\Cerny's conjecture states that $\rt \leq (|Q|-1)^2$ holds for every synchronizing automaton~$A$.
It is named in honor of foundational work by \Cerny, see~\cite{cerny-translation} for a translation into English, that established the corresponding lower bound via a series of synchronizing automata.
Remarkably, this seemingly simple conjecture has resisted attempts at a general solution for decades, although it has been proven for special classes of automata, e.g., Eulerian automata~\cite{sync-eulerian}, monotonic automata~\cite{eppstein-reset}, and others~\cite{volkov-survey22}.
Since an improvement of the asymptotic bound $\bigO(|Q|^3)$ has proven elusive, researchers have looked at the coefficient $\alpha$ of the cubic term.
The current best value is $\alpha \leq 0.1654$ due to Shitov~\cite{upper-bound-shitov}, who expanded on ideas by Szyku{\l}a~\cite{upper-bound-szykula}.
Kiefer and Ryzhikov~\cite{minimize-rank} recently improved the running time of an algorithm for the more general problem of finding a word $w$ that minimizes $|Q \cdot w|$ (called the \emph{rank} of $w$).

A lower bound on the achievable approximation ratio for $\rt$ was improved several times~\cite{inapproximability-binary,inapproximability-log,inapproximability-scs}, until Gawrychowski and Straszak~\cite{strong-inapproximability} showed that it is \NP-hard to approximate $\rt$ within a factor $|Q|^{1-\varepsilon}$ for every $\varepsilon > 0$.
On the practical side, algorithm engineering efforts have resulted in implementations that find shortest synchronizing words for instances with more than 500 states~\cite{shortest-sync-algeng}.

Deciding if a subset $S \subseteq Q$ is synchronizing is \PSPACE-complete, as stated by Natarajan~\cite{natarajan-orienters} and formally proven by Rystsov~\cite{automata-rystsov83}.
This problem is strongly related to \textsc{IntersectionNonEmptiness}: Given a family of automata over the same alphabet, decide whether they accept a common word.
Kozen~\cite{intersection-kozen} originally proved this problem \PSPACE-complete, and it remains \NP-hard for tally automata~\cite{intersection-timecop,intersection-unary-stoc73}, i.e., automata over a unary alphabet.

\subsection{Efficient gathering in the full tilt model}

By \cref{obs:gathering-localization}, a gathering sequence for a polyomino $P$ corresponds precisely to a synchronizing word for the automaton $(V,\D,\delta^1)$.
Thus, we could solve \fullGathering by applying Eppstein's algorithms~\cite{eppstein-reset} to $(V,\D,\delta^1)$ to decide if it is synchronizing in $\bigO(N^2)$ time and, if so, find a gathering sequence for $P$ in $\bigO(N^3)$ time.
We now sketch algorithms with running times of $\bigO(n^2)$ and $\bigO(n_c n^2)$, respectively, that use only $\partial P$ as input.
This can make a huge difference as the number of pixels $N$ might not be polynomial in the~number~of~corners~$n$.

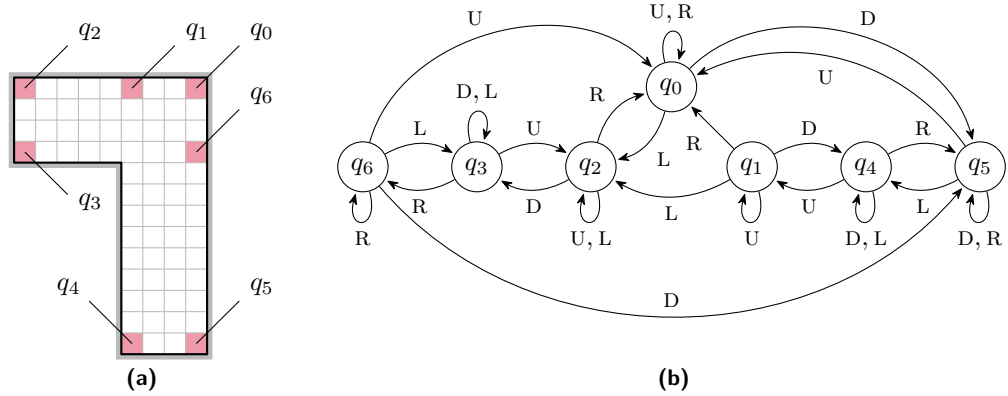
\begin{figure}[tbh]
  \begin{subcaptionblock}{0.3\textwidth}
    \phantomcaption\label{sfig:ft-automaton-a}
    \centering
    \includegraphics{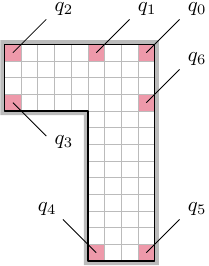}\\
    \captiontext*{}
  \end{subcaptionblock}\hfill
  \begin{subcaptionblock}{0.7\textwidth}
    \phantomcaption\label{sfig:ft-automaton-b}
    \centering
    \begin{tikzpicture}[node distance=1.5cm,on grid,auto,>={Stealth[round]},shorten >=1pt,
                        every state/.style={minimum size=0cm,inner sep=0.1cm}]
      \node[state] (q_6)                      {$q_6$};
      \node[state] (q_3) [right=of q_6]       {$q_3$};
      \node[state] (q_2) [right=of q_3]       {$q_2$};
      \node[state] (q_0) [above right=of q_2] {$q_0$};
      \node[state] (q_1) [below right=of q_0] {$q_1$};
      \node[state] (q_4) [right=of q_1]       {$q_4$};
      \node[state] (q_5) [right=of q_4]       {$q_5$};

      \path[->] (q_6) edge [bend right=50] node {$\ddir$}       (q_5)
                      edge [bend left]     node {$\ldir$}       (q_3)
                      edge [bend left=60]  node {$\udir$}       (q_0)
                      edge [loop below]    node {$\rdir$}       ()
                (q_3) edge [bend left]     node {$\rdir$}       (q_6)
                      edge [bend left]     node {$\udir$}       (q_2)
                      edge [loop above]    node {$\ddir,\ldir$} ()
                (q_2) edge [bend left]     node {$\ddir$}       (q_3)
                      edge [bend left]     node {$\rdir$}       (q_0)
                      edge [loop below]    node {$\udir,\ldir$} ()
                (q_0) edge [bend left]     node {$\ldir$}       (q_2)
                      edge [bend left=60]  node {$\ddir$}       (q_5)
                      edge [loop above]    node {$\udir,\rdir$} ()
                (q_1) edge                 node {$\rdir$}       (q_0)
                      edge [bend left]     node {$\ldir$}       (q_2)
                      edge [bend left]     node {$\ddir$}       (q_4)
                      edge [loop below]    node {$\udir$}       ()
                (q_4) edge [bend left]     node {$\udir$}       (q_1)
                      edge [bend left]     node {$\rdir$}       (q_5)
                      edge [loop below]    node {$\ddir,\ldir$} ()
                (q_5) edge [bend left]     node {$\ldir$}       (q_4)
                      edge [bend right=40] node {$\udir$}       (q_0)
                      edge [loop below]    node {$\ddir,\rdir$} ();
    \end{tikzpicture}\\
    \captiontext*{}
  \end{subcaptionblock}\hfill
  \caption{\subref{sfig:ft-automaton-a}~A polyomino $P$ and its significant pixels. \subref{sfig:ft-automaton-b}~The full tilt automaton $A(P)$.}\label{fig:ft-automaton}
\end{figure}

The key observation is that $\delta^\cup$ very quickly reduces the set of occupied pixels to one of size~$\bigO(n)$.
Consider the set $S \subseteq V$ of \emph{significant pixels} that contains all corner pixels and possibly \emph{helper pixels}, which are found by extending sides of $\partial P$ that meet at reflex corners until they hit another side of $\partial P$.
\Cref{sfig:ft-automaton-a} illustrates the significant pixels of a polyomino~$P$ with helper pixels $q_1$ and $q_6$.
Observe that $|S| \leq 2n-n_c$ because at most two helper pixels are included per reflex corner.
Any sequence of two moves in perpendicular directions changes the position of a particle from any pixel $p \in V$ to one in $S$, e.g., $\delta^1(\ddir\ldir)(p) \in S$.
$S$ is closed under $\delta^1$, i.e., after the move sequence $\ddir\ldir$ we can focus on $\delta^1$ restricted to~$S$.
We call the resulting automaton $A(P)=(S,\D,\delta^1)$ the \emph{full tilt automaton} of $P$, see \cref{sfig:ft-automaton-b}.

\begin{lemma}\label{lem:gather-sync-pa}
A polyomino $P$ is gatherable if and only if the full tilt automaton $A(P)$ is synchronizing, in which case $\rt[A(P)] \leq \sgs \leq \rt[A(P)] + 1$ holds.
\end{lemma}
\begin{proof}
A gathering sequence for $P$ is a synchronizing word for $A(P)$ because $S \subseteq V$.
Conversely, assume $w$ is a synchronizing word for $S$.
If $w=\varepsilon$, then $P$ is a single pixel and $\varepsilon$ a gathering sequence.
Otherwise, assume without loss of generality that $w=\ldir{}w'$ for some $w' \in \D^*$ (rotate $P$ if necessary).
Then $\ddir\ldir{}w = \ddir\ldir\ldir{}w'$ is a gathering sequence for $P$ and can be shortened to $\ddir\ldir{}w' = \ddir{}w$, because moves in the full tilt model are idempotent.
\end{proof}

It is a well-known fact~\cite{natarajan-orienters,volkov-survey22} that an automaton $(Q,\Sigma,\delta)$, with $|Q| > 1$, is synchronizing if and only if every unordered pair of states $\{s,t\} \subseteq Q$ has a synchronizing word.
This suggests a strategy of iteratively merging pairs, which is also the foundation of known strategies for gathering in \SSt~\cite{gathering-icra20,gathering-case16} and for robot localization~\cite{localization-icra20,localization-tor07}.
To efficiently determine synchronizing words for pairs in $A(P)$, we construct the \emph{pair automaton} ${A(P)^{\leq 2}=(S^{\leq 2},\D,\delta^{\leq 2})}$, where $S^{\leq 2} = \{\{p,q\}: p,q \in S\}$ and $\delta^{\leq 2}(v)(\{p,q\}) = \{\delta^1(v)(p)\} \cup \{\delta^1(v)(q)\}$.
Note that~$S^{\leq 2}$ also contains all singletons $\{p\} \subseteq S$.
A word $w$ with $p \cdot w = q \cdot w$ in $A(P)$ corresponds precisely to a path labeled $w$ from $\{p,q\}$ to $\{p \cdot w\}$ in the digraph underlying $A(P)^{\leq 2}$.
Using breadth-first search on this digraph to build a shortest-path forest allows to pre-compute the next step on a shortest path to a singleton for every pair of states, if such a path exists.

To cut down the number of pairs that need to be considered, we initially reduce $S$ to corner pixels only by repeatedly applying two perpendicular moves, e.g., $\ddir\ldir$.
Observe that if $p \in S$ is a helper pixel, then $p \cdot \ddir\ldir < p$ in terms of their lexicographic order, as at least one coordinate decreases.
Thus, $\bigO(n)$ applications of $\ddir\ldir$ suffice to eliminate all helper pixels (and all corner pixels not at lower-left corners) from $S$.
This way we need to check at most $\bigO(n_c)$ pairs, improving the running time for non-simple polyominoes with many reflex corners.

\begin{algorithm2e}[tbh]
  \DontPrintSemicolon
  \KwIn{The boundary $\partial P$ of a polyomino $P$.}
  \KwOut{A gathering sequence for $P$ in \FT, if one exists.}
  Compute the full tilt automaton $A(P) = (S,\D,\delta^1_{\FT})$.\;
  Construct the pair automaton $A(P)^{\leq 2}$.\;
  Use BFS to build a shortest-path forest in the digraph underlying $A(P)^{\leq 2}$.\;
  $w \leftarrow \ddir\ldir$\;
  $X \leftarrow S$\;
  \While{$X$ contains helper pixels\label{line:preprocess-start}}
        {$w \leftarrow w\ddir\ldir$\;
         $X \leftarrow \bigcup_{p \in X} p \cdot \ddir\ldir$\label{line:preprocess-end}\;
        }
  \While{$|X| > 1$\label{line:main-loop-start}}
        {Select $p,q \in X$ with $p \neq q$.\;
         \eIf{there is $w' \in \D^*$ such that $|\{p,q\} \cdot w'| = 1$ in $A(P)^{\leq 2}$}
             {$w \leftarrow ww'$\;
              $X \leftarrow \bigcup_{p \in X} p \cdot w'$\;
             }
             {\KwRet{$P$ is not gatherable.}\;}
         }
  \KwRet{w}\;
  \caption{\fullGathering in \FT.}\label{alg:full-gather}
\end{algorithm2e}

\begin{theorem}\label{thm:alg-gathering}
Given the boundary $\partial P$ of a polyomino~$P$ with $n_c \leq n$ convex corners, a gathering sequence for $P$, if one exists, can be computed in~$\bigO(n_c n^2)$ time.
\end{theorem}
\begin{proof}
The steps to compute a gathering sequence, if one exists, are listed in \cref{alg:full-gather}.
Correctness follows from \cref{lem:gather-sync-pa} and the well-known strategy to compute synchronizing words by repeatedly merging two states~\cite{natarajan-orienters,volkov-survey22}.
As to the running time, $A(P)$ can be constructed in $\bigO(n \log n)$ time using a data structure by Sarnak and Tarjan~\cite{planar-point-location} for planar point location, allowing us to pre-process $\partial P$ in $\bigO(n \log n)$ time such that afterwards we can find, for any point $p$ inside $P$, the side of $\partial P$ intersecting the projection of $p$ to the left, to the right, up, or down in $\bigO(\log n)$ time.
Thus, we can first find the helper pixels and then compute the functions $\delta^1(v)$, $v \in \D$, on every significant pixel in total time $\bigO(n \log n)$.
Afterwards, $A(P)^{\leq 2}$ can be computed in $\bigO(n^2)$ time.
Applying breadth-first search to $A(P)^{\leq 2}$ takes time linear in the size of $A(P)^{\leq 2}$, i.e., $\bigO(n^2)$ time.
The pre-processing loop in lines~\ref{line:preprocess-start}--\ref{line:preprocess-end} runs for $\bigO(n)$ iterations of $\bigO(n)$ time each and ensures that afterwards $|X| \leq n_c$.
Thus, the main loop starting with line~\ref{line:main-loop-start} is executed at most $\bigO(n_c)$ times.
We refer to Eppstein~\cite{eppstein-reset} for techniques that make every iteration of the main loop run in $\bigO(n^2)$ time, for a total running time of~$\bigO(n_c n^2)$.
\end{proof}

If we are only interested in the (non-)existence of a gathering sequence, we can stop after computing the shortest-path forest.
A gathering sequence exists if and only if every pair of states has been assigned a successor on a path to a singleton~\cite{eppstein-reset}.

\begin{corollary}\label{cor:alg-gathering-check}
Given the boundary $\partial P$ of a polyomino~$P$ with $n$ corners, one can decide in $\bigO(n^2)$ time whether $P$ is gatherable.
\end{corollary}

The idea underlying the pair automaton reveals initial insight into \paraGathering when parameterized by the number of particles $k$.
A generalized $k$-subset automaton $A(P)^{\leq k}$, with all subsets of pixels of size at most $k$ as states, can be used to find shortest gathering sequences in $\bigO(N^k)$ time.
Thus, \paraGathering is in \XP when parameterized by $k$ and the dual graph is given as input.

\subsection{Worst-case bounds for gathering sequences in {\boldmath \FT}}

Once a polyomino is known to be gatherable, how long can a shortest gathering sequence be in the worst case?
An upper bound matches the running time of the algorithm.

\begin{corollary}\label{cor:gather-upper-bound}
For every gatherable polyomino $P$, $\sgs \in \bigO(n_c n^2)$.
\end{corollary}

The basic idea for a lower bound, depicted in \cref{fig:gathering-lower-bound}, is to place two particles on a directed cycle of length $2m+1$, spaced apart by $m$ positions.
Moving against the direction of the cycle leads to dead ends, except at $p_{2m-1}$.
There, $\ldir\udir$ allows to skip one step on the cycle, which can be used to reduce the distance between the particles by $1$.
To reduce the distance again, the whole cycle has to be traversed, resulting in $\Omega(m)$ iterations of $\Omega(m)$ moves to gather the particles, for a total of $\Omega(m^2) = \Omega(n^2)$.

\begin{figure}[tbh]
  \centering
  \includegraphics{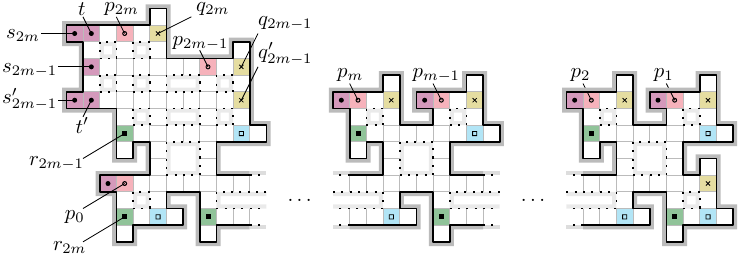}
  \caption{A family of polyominoes $P_m$ with $\sgs[P_m] \in \Omega(n^2)$ in \FT. Important classes of pixels are marked with distinct colors and symbols. Dotted areas indicate optional holes that can turn the simple polyomino into a maze, ensuring the construction works for both special cases.}\label{fig:gathering-lower-bound}
\end{figure}

\begin{theorem}\label{thm:gather-lower-bound}
There is an infinite family $P_m$, $m \in \N$, of gatherable simple polyominoes such that $\sgs[P_m] \in \Omega(n^2)$, where the number of corners $n$ grows with $m$.
\end{theorem}
\begin{proof}
Pick $m \in \N$ and let $M=2m+1$.
\Cref{fig:gathering-lower-bound} depicts a polyomino~$P_m$, which consists of one unique gadget to the left followed by $m-1$ repetitions of a simpler gadget, all connected sequentially at the bottom.
Several classes of pixels are marked with distinct colors and symbols; the most important ones are $V_p$, $V_q$, $V_r$, and $V_s$.
Each class $V_x$ contains pixels indexed $x_0$ to $x_{2m}$, with $V_q$ containing the extra pixel $q'_{2m-1}$ and $V_s$ the additional pixels $s'_{2m-1}$, $t$, and $t'$.
These classes are easily verified to form a congruence with respect to $\delta^1(v)$ for all $v \in \D$, e.g., $\delta^1(\udir)(r_i) = p_{(i+1) \bmod M}$ and $\delta^1(\ddir)(s) \in V_s$ for all $s \in V_s$.
There are more congruence classes, in particular the remaining pixels of degree $1$; they are implicitly indexed like their unique neighbors but remain unmarked to avoid clutter.
Note that the union of these classes is closed under the application of moves $\delta^1(v)$.

Particles can move from $p_i$ to $p_{(i+1) \bmod M}$ via the sequence $\rdir\ddir\ldir\udir$.
We define the \emph{cyclic index} $\operatorname{idx}(x_i) = i$ for a pixel $x_i$ with index $i$, and set $\operatorname{idx}(t) = 2m$ and $\operatorname{idx}(t') = 2m-1$, thereby extending $\operatorname{idx}$ to a total function on pixels in the relevant classes.
The \emph{cyclic distance} of two pixels $a,b$ is defined as $\Delta(a,b) = \min\{|\operatorname{idx}(a) - \operatorname{idx}(b)|, M - |\operatorname{idx}(a) - \operatorname{idx}(b)|\}$.
Clearly, a gathering sequence $w \in \D^*$ for a pair of pixels $\{a,b\}$ satisfies $\Delta(a \cdot w, b \cdot w) = 0$.
We now show that any sequence $w \in \D^*$ with $\Delta(p_0 \cdot w, p_m \cdot w) = 0$ has length $|w| \in \Omega(m^2)$.

Consider two particles initially placed at $p_0$ and $p_m$ with cyclic distance $\Delta$, i.e., $\Delta = \Delta(p_0, p_m) = m$ at the outset.
The congruence ensures that both particles are always in the same class.
By construction, all moves $v$ satisfy $\operatorname{idx}(a \cdot v) \in \{\operatorname{idx}(a), (\operatorname{idx}(a)+1) \bmod M\}$ for all pixels $a$.
Therefore, to change $\Delta$, we need two pixels $a,b$ in the same class such that $\operatorname{idx}(a \cdot v) = \operatorname{idx}(a)$ while $\operatorname{idx}(b \cdot v) = (\operatorname{idx}(b)+1) \bmod M$ for some $v \in \D$.
There is exactly one place in $P_m$ where this can happen: One particle must be moved to $t$ via $\udir$, while the other remains at some $s_i$.
This changes $\Delta$ from $\Delta(a,s_i)$ to $\Delta(t,s_i) = \Delta(a,s_i) \pm 1$, where $a \in \{t',s_{2m-1}\}$.
After every decrease of $\Delta$, the particle located at $t$ has to traverse the whole cycle through $p_0,p_1, \ldots, p_{2m-1}$ before $\Delta$ can be decreased again via another move to~$t$.
Therefore, reducing $\Delta$ to $0$ requires $m$ traversals of $\Omega(m)$ moves, for a total of $\Omega(m^2)$.

To see that $P_m$ is gatherable, observe that $C \cdot \rdir\ddir\ldir\udir \subseteq V_p$, for all configurations~$C$.
Afterwards, we can successively merge pairs in $V_p$ by moving along the cycle to $p_{2m-1}$ and reducing their cyclic distance via $p_{2m-1} \cdot \ldir\udir = t$.
Any gathering sequence $w$ for $P_m$ must be a gathering sequence for $\{p_0,p_m\}$ and as such satisfy $\Delta(p_0 \cdot w,p_m \cdot w) = 0$.
Since the construction guarantees $m \in \Omega(n)$, this implies $\sgs[P_m] \in \Omega(n^2)$.
\end{proof}

\section{Simulating binary automata with uniformly controlled particles}\label{sec:simulation}

If \Cerny's conjecture is true, then the $\Omega(n^2)$ lower bound on $\sgs$ is actually tight.

\begin{conjecture}\label{conj:gathering-upper-bound}
For every gatherable polyomino $P$ with $n$ corners, $\sgs \in \bigO(n^2)$.
\end{conjecture}

The geometric structure of polyominoes should make \cref{conj:gathering-upper-bound} easier to prove than \Cerny's conjecture.
Alas, we must dampen this hope, because particles in polyominoes can be used to simulate binary automata, i.e., automata with alphabet $\Sigma = \{0,1\}$.

\begin{figure}[tbh]
  \centering
  \includegraphics{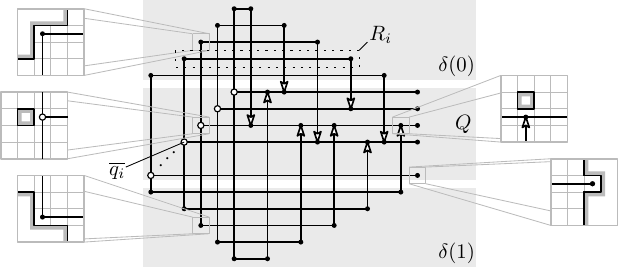}
  \caption{Schematic overview of the construction of $P(A)$ from a binary automaton $A$. Three blocks represent, from top to bottom, transitions for~$0$, states $q_i \in Q$, and transitions for~$1$. We zoom in to show how the schema is realized by the boundary $\partial P(A)$. See \cref{fig:simulation-example} for a specific example.}\label{fig:simulation-scheme}
\end{figure}

To simulate an automaton $A=(Q,\{0,1\},\delta)$, the main idea is to build a polyomino $P(A)$ that has every state $q \in Q$ represented by a pixel $\overline{q}$, and use fixed sequences of moves to embody transitions.
A single particle is placed inside $P(A)$, its presence at $\overline{q}$ signifying that $A$ is in state $q$.
We use the clockwise cycle $\udir\rdir\ddir\ldir$ and the counter-clockwise cycle $\ddir\rdir\udir\ldir$ to express transitions for $0$ and $1$, respectively.
By carefully constructing the boundary of $P(A)$, we guarantee that all minimal sequences of moves between two pixels $\overline{q_i}$ and $\overline{q_j}$ can be decomposed into combinations of $\udir\rdir\ddir\ldir$ and $\ddir\rdir\udir\ldir$, thus ensuring that minimal sequences of moves between pixels representing states are in correspondence with transitions in $A$.

Without loss of generality, let $Q=\{q_0,q_1,\ldots,q_{|Q|-1}\}$, and set $\overline{Q'} = \{\overline{q}: q \in Q'\}$ for sets $Q' \subseteq Q$.
Our construction of $P(A)$ is schematically illustrated in \cref{fig:simulation-scheme}.
It consists of three vertically arranged blocks, the middle one emulating $Q$.
Every state $q_i \in Q$ gets assigned a row of width $6|Q|$ including the representative pixel $\overline{q_i}$.
These rows are stacked, with left-aligned rows of width $6|Q|-1$ in between, and the $\overline{q_i}$ are aligned on an upward diagonal, from $\overline{q_{|Q|-1}}$ in the bottom left up to $\overline{q_0}$.
A $1 \times 1$ hole is placed to the left of every $\overline{q_i}$ except $\overline{q_{|Q|-1}}$, which is next to the outer boundary.

The top and bottom blocks embody the transitions for $0$ and $1$, respectively.
To implement the transitions $\delta(0)(q_i)$ in the top block, rows $R_i$ of width $6i+3$ are stacked from top to bottom and aligned on the left with $\overline{q_i}$.
A $1 \times 1$ hole is placed below the pixel at the intersection of the rightmost column of $R_i$ and the row of $\overline{q_j}$ when $q_j = \delta(0)(q_i)$.
This allows us to simulate a transition for $0$ by the sequence of moves $\udir\rdir\ddir\ldir$.
Similar rows and holes are constructed for $\delta(1)(q_i)$ in the bottom block, except now the rows are arranged from bottom to top and of width $6(|Q|-i)-1$, allowing us to simulate $1$ by $\ddir\rdir\udir\ldir$.
As a minor detail, note that the hole placed to realize $\delta(1)(q_{|Q|-1})$ ends up connected to the outer face, i.e., $\delta(1)(q_{|Q|-1})$ does not actually add another hole.
The full construction requires $6|Q|$ convex corners and $6(3|Q|-2)$ reflex corners, for a total of $n=12(2|Q|-1)$, and involves $N \in \bigO(|Q|^2)$ pixels.
Thus, both $\partial P(A)$ and $G_{P(A)}$ can be computed in polynomial time.

\begin{figure}[tbh]
  \begin{subcaptionblock}{0.16\textwidth}
    \phantomcaption\label{sfig:simulation-example-auto}
    \centering
    \begin{tikzpicture}[node distance=1.5cm,on grid,auto,>={Stealth[round]},shorten >=1pt,
                        every state/.style={minimum size=0cm,inner sep=0.1cm}]
      \node[state,initial above,initial text=] (q_0)                      {$q_0$};
      \node[state,accepting]                   (q_1) [below=of q_0]       {$q_1$};
      \node[state]                             (q_2) [below=of q_1]       {$q_2$};

      \path[->] (q_0) edge              node {$0$}   (q_1)
                      edge [loop right] node {$1$}   ()
                (q_1) edge              node {$0$}   (q_2)
                      edge [loop right] node {$1$}   ()
                (q_2) edge [loop below] node {$0,1$} ();
    \end{tikzpicture}\\
    \captiontext*{}
  \end{subcaptionblock}\hfill
  \begin{subcaptionblock}{0.42\textwidth}
    \phantomcaption\label{sfig:simulation-example-poly}
    \centering
    \includegraphics[page=2]{auto2poly-example}\\
    \captiontext*{}
  \end{subcaptionblock}\hfill
  \begin{subcaptionblock}{0.42\textwidth}
    \phantomcaption\label{sfig:simulation-example-maze}
    \centering
    \includegraphics[page=1]{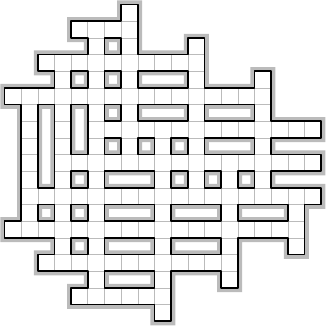}\\
    \captiontext*{}
  \end{subcaptionblock}\hfill
  \caption{\subref{sfig:simulation-example-auto}~The binary automaton $A_0$ that accepts all words containing exactly one $0$. \subref{sfig:simulation-example-poly}~The polyomino $P(A_0)$. Colors and symbols indicate congruence classes of pixels reachable from representatives $\overline{q_i}$. \subref{sfig:simulation-example-maze}~$P(A_0)$ turned into a maze $M(A_0)$ by adding holes and pixels of degree $1$.
}\label{fig:simulation-example}
\end{figure}

For a specific example consider the binary automaton $A_0$ that accepts the words containing exactly one $0$, see \cref{sfig:simulation-example-auto}.
The polyomino $P(A_0)$ resulting from our construction is depicted in \cref{sfig:simulation-example-poly}, where marked pixels show congruence classes of pixels reachable from $\overline{Q}$, in particular those traversed on cycles $\udir\rdir\ddir\ldir$ and $\ddir\rdir\udir\ldir$.

\begin{observation}\label{obs:simulation-isomorphism}
For all $p,q \in Q$, $p \xrightarrow{0} q \Leftrightarrow \overline{p} \xrightarrow{\udir\rdir\ddir\ldir} \overline{q}$ and $p \xrightarrow{1} q \Leftrightarrow \overline{p} \xrightarrow{\ddir\rdir\udir\ldir} \overline{q}$.
\end{observation}

\begin{lemma}\label{lem:simulation-shorten}
For all $p \in Q$, with representative pixels $\overline{p} \in P(A)$, and all $w \in \D^*$~with $\overline{p} \cdot w \in \overline{Q}$, there is $w' \in \{\udir\rdir\ddir\ldir, \ddir\rdir\udir\ldir\}^*$ with $|w'| \leq |w|$ that satisfies $\overline{q} \cdot w' = \overline{q} \cdot w$ for~all~$q \in Q$.
\end{lemma}
\begin{proof}
Let $w = w_1 w_2 \ldots w_{|w|}$, $w_i \in \D$, such that $\overline{p} \cdot w \in \overline{Q}$ for some $p \in Q$.
Furthermore, let $p_i = \overline{p} \cdot (w_1 w_2 \ldots w_i)$, and set $p_0 = \overline{p}$.
The proof is by induction on $\lambda = |\{i: {0 \leq i \leq |w|},\allowbreak {p_i \in \overline{Q}}\}|$.
If $\lambda = 1$, then $w = \varepsilon$ and we can choose $w' = \varepsilon$.
Otherwise, split the sequence at the next pixel from $\overline{Q}$ by choosing the largest $\ell$ such that $\{p_1,p_2,\ldots,p_{\ell-1}\} \cap \overline{Q} = \varnothing$ and setting $w_{\text{start}} = w_1 w_2 \ldots w_\ell$ and $w_{\text{rest}} = w_{\ell+1} \ldots w_{|w|}$.
Then, $\ell \geq 4$, as there are no shorter (non-empty) paths between pixels in $\overline{Q}$, and $w_\ell = \ldir$, because that is the only way to reach $\overline{Q}$.
Now, $p_{\ell-1}$ must be one of $\overline{p} \cdot \udir\rdir\ddir$, $\overline{p} \cdot \ddir\rdir\udir$, and $\overline{p} \cdot \rdir$.
Consider the largest index $\ell' \leq \ell$ such that $p_{\ell'} \neq \overline{p} \cdot \rdir$.
Then, $p_{\ell'}$ is either $\overline{p} \cdot \udir\rdir\ddir$ or $\overline{p} \cdot \ddir\rdir\udir$.
In the first case, set $w_{\text{start}}' = \udir\rdir\ddir\ldir$, otherwise $w_{\text{start}}' = \ddir\rdir\udir\ldir$.
The classification of reachable pixels into the congruence classes depicted in \cref{sfig:simulation-example-poly} ensures that $w_{\text{start}}'$ is equivalent to $w_{\text{start}}$ on $\overline{Q}$.
Apply the induction hypothesis to $w_{\text{rest}}$ to get an equivalent $w_{\text{rest}}' \in \{\udir\rdir\ddir\ldir, \ddir\rdir\udir\ldir\}^*$ with $|w_{\text{rest}}'| \leq |w_{\text{rest}}|$.
The concatenation $w' = w_{\text{start}}' w_{\text{rest}}'$ now fulfills the requirements.
\end{proof}

Having established these observations, we can now see how the polyomino $P(A)$ captures properties of the automaton $A$.
In particular, the length of a shortest gathering sequence for a subset of representative pixels $\overline{S}$ (if one exists) is exactly four times as long as a shortest synchronizing word for the subset of states $S$.
Furthermore, a gathering sequence for $P(A)$ exists if and only if one exists for the complete set of representative pixels $\overline{Q}$, and it requires at most five additional moves.

\begin{theorem}\label{thm:simulation-acceptance}
A word $w \in \{0,1\}^*$ is accepted by an automaton $A=(Q,\{0,1\},\delta,q_0,F)$ if and only if there is a sequence of moves $\overline{w} \in \D^*$ such that $\overline{q_0} \cdot \overline{w} \subseteq \overline{F}$ in $P(A)$.
\end{theorem}
\begin{proof}
For the forward direction, apply \cref{obs:simulation-isomorphism}. For the backward direction, first apply \cref{lem:simulation-shorten} to get $\overline{w}' \in \{\udir\rdir\ddir\ldir, \ddir\rdir\udir\ldir\}^*$, then apply \cref{obs:simulation-isomorphism} to $\overline{w}'$.
\end{proof}

\begin{theorem}\label{thm:subset-sync-gather}
A set $S \subseteq Q$ of states of an automaton $A=(Q,\{0,1\},\delta)$ is synchronizing if and only if its set of representatives $\overline{S}$ is gatherable in $P(A)$, in which case $\sgs[\overline{S}] = 4 \rt[S]$.
\end{theorem}
\begin{proof}
For the forward direction, let $w \in \{0,1\}^*$ be a synchronizing word for $S$ of minimum length.
Use \cref{obs:simulation-isomorphism} to build a gathering sequence $\overline{w} \in \{\udir\rdir\ddir\ldir, \ddir\rdir\udir\ldir\}^*$ for $\overline{S}$ satisfying $|\overline{w}| = 4 |w|$.
Assume, for sake of contradiction, that there is a gathering sequence $\overline{w}' \in \D^*$ with $|\overline{w}'| < |\overline{w}|$.
First apply \cref{lem:simulation-shorten} to $\overline{w}'$ and then \cref{obs:simulation-isomorphism} to get a synchronizing word $w' \in \{0,1\}^*$ for $S$ with $|w'| < |w|$, contradicting the assumption that $w$ has minimum length.
The backward direction is very similar.
\end{proof}

\begin{corollary}\label{cor:sync-gather}
An automaton $A=(Q,\{0,1\},\delta)$ is synchronizing if and only if $P(A)$ is gatherable, in which case $4 \rt \leq \sgs[P(A)] \leq 4 \rt + 5$ holds.
\end{corollary}
\begin{proof}
Observe that $V \cdot \rdir\ddir\rdir\udir\ldir \subseteq \overline{Q}$, i.e., $P(A)$ is gatherable if and only if $\overline{Q}$ is, with $\sgs[\overline{Q}] \leq \sgs[P(A)] \leq \sgs[\overline{Q}] + 5$.
The claim now follows from \cref{thm:subset-sync-gather} for $S=Q$.
\end{proof}

\begin{corollary}\label{cor:weak-cerny}
If $\sgs \in \bigO(n^2)$ for every gatherable polyomino $P$ with $n$ corners, then every synchronizing binary automaton $A=(Q,\{0,1\},\delta)$ has $\rt \in \bigO(|Q|^2)$.
\end{corollary}
\begin{proof}
Let $A$ be a synchronizing binary automaton.
Then $P(A)$ is gatherable, by \cref{cor:sync-gather}, and $\sgs[P(A)] \in \bigO(n^2)$, by this claim's hypothesis.
Now, $\sgs[P(A)] \in \bigO(|Q|^2)$, since $n=12(2|Q|-1)$.
By \cref{cor:sync-gather}, $4 \rt \leq \sgs[P(A)]$, i.e., $\rt \in \bigO(|Q|^2)$.
\end{proof}

While a quadratic upper bound on the reset threshold for synchronizing binary automata would not settle \Cerny's conjecture, it would be remarkable progress.
After decades of research, the existence of such an upper bound is still an open problem, suggesting that a proof for \cref{conj:gathering-upper-bound} may be hard to find.

It is easy to adapt the construction of $P(A)$ to produce a maze $M(A)$, at the cost of possibly requiring $\Omega(|Q|^2)$ corners, by adding more holes as well as pixels of degree $1$ next to corner pixels, as illustrated for $M(A_0)$ in \cref{sfig:simulation-example-maze}.

\begin{corollary}\label{cor:subset-gathering-mazes}
\emph{\subsetGathering} is \PSPACE-complete, even in mazes.
\end{corollary}
\begin{proof}
It is easy to see that a non-deterministic algorithm for \subsetGathering could run in polynomial space, which implies containment in $\NPSPACE = \PSPACE$~\cite{savitch}.

To prove \PSPACE-hardness, we reduce from \textsc{SubsetSynchronizability}: Given an automaton $A=(Q,\Sigma,\delta)$ and a set $S \subseteq Q$, determine if $S$ is synchronizing.
This problem was originally shown \PSPACE-hard by Rystsov~\cite{automata-rystsov83}, and Vorel~\cite{subset-sync-lower-bound} proved it remains so for $\Sigma = \{0,1\}$.
We construct $P(A)$ for a given binary automaton $A=(Q,\{0,1\},\delta)$ in polynomial time.
By \cref{thm:subset-sync-gather}, $S \subseteq Q$ is synchronizing if and only if $\overline{S}$ is gatherable.
This property remains true when $P(A)$ is turned into a maze $M(A)$.
\end{proof}

We want to note that \PSPACE-hardness of \subsetGathering in general polyominoes also follows easily from the \PSPACE-hardness proof for \occupancy from~\cite{tilt-soda20}, but we welcome having an alternative proof in keeping with our automata theme.

\begin{corollary}\label{cor:quick-gathering-inapprox}
Unless $\P = \NP$, there is no polynomial-time approximation algorithm for \emph{\quickGathering} with approximation ratio $n^{1-\varepsilon}$, for any $\varepsilon > 0$.
\end{corollary}
\begin{proof}
Let $A$ be a synchronizing binary automaton.
By \cref{cor:sync-gather}, a polynomial time $n^{1-\varepsilon}$\nobreakdash-approximation for $\sgs[P(A)]$ could be turned into a polynomial-time $|Q|^{1-\varepsilon'}$\nobreakdash-\hspace{0pt}approximation for $\rt$, $0 < \varepsilon' < \varepsilon$.
Gawrychowski and Straszak~\cite{strong-inapproximability} showed that the existence of such an approximation algorithm implies $\P = \NP$.
\end{proof}

Note that, since $N \in \bigO(|Q|^2)$ holds even for $M(A)$, the reasoning used in \cref{cor:quick-gathering-inapprox} also excludes the possibility of a polynomial-time $N^{\sfrac{1}{2}-\varepsilon}$-approximation for mazes (unless~$\P = \NP$).

\begin{corollary}\label{cor:subset-gathering-lower-bound}
There are infinitely many polyominoes $P$ with gatherable configurations $C$ such that $\sgs[C] \in 2^{\Omega(n)}$, where $n$ is the number of corners of~$P$.
\end{corollary}
\begin{proof}
Vorel~\cite{subset-sync-lower-bound} showed that there are infinitely many binary automata $A=(Q,\{0,1\},\delta)$ and synchronizing sets $S \subseteq Q$ such that $\rt[S] \geq 2^\frac{|Q|}{21}$.
Thus, by \cref{thm:subset-sync-gather}, constructing $P(A)$ yields polyominoes (with $n < 24|Q|$ corners) and gatherable configurations $\overline{S}$ such that $\sgs[\overline{S}] \geq 4 \rt[S] > 2^{\frac{n}{21 \cdot 24}} = 2^{\frac{n}{504}} \in 2^{\Omega(n)}$.
\end{proof}

\section{Approximating gathering sequences in simple polyominoes}\label{sec:optimization}

We first show that \quickGathering{} is \NP-hard even in simple mazes.
To do this, we reduce from \textsc{ShortestCommonSupersequence} where, given an alphabet~$\Sigma$ and a set $S$ of words over $\Sigma$, we want to find a shortest word $\overline{w}$ such that for each $w \in S$, there exists a sequence $i_1 < \dots < i_{|w|}$ with $w_j = \overline{w}_{i_j}$ for each~$j$, i.e., $\overline{w}$ is a supersequence of each word $w \in S$.
This problem is already \NP-hard if $\Sigma = \{0,1\}$~\cite{scs-binary}.
Gerbush and Heeringa~\cite{inapproximability-scs} used a reduction from a generalized problem to prove inapproximability of the reset threshold of synchronizing automata, which inspired our approach.

\begin{figure}[tb]
    \begin{subcaptionblock}{0.22\textwidth}%
        \centering%
        \phantomcaption\label{sfig:scs-reduction-binary-gadget-a}%
        \includegraphics[page=1]{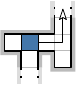}%

        \captiontext*{}%
        \phantomcaption\label{sfig:scs-reduction-binary-gadget-b}%
        \includegraphics[page=2]{scs-reduction-binary}%

        \captiontext*{}%
    \end{subcaptionblock}%
    \hfill%
    \begin{subcaptionblock}{0.75\textwidth}%
        \phantomcaption\label{sfig:sgs-reduction-binary-example}%
        \centering%
        \includegraphics[page=3]{scs-reduction-binary}%

        \captiontext*{}%
    \end{subcaptionblock}%
    \caption{%
        \subref{sfig:scs-reduction-binary-gadget-a}~$0$-gadget, move sequence $\rdir\udir$.
        \subref{sfig:scs-reduction-binary-gadget-b}~$1$-gadget, move sequence $\ldir\udir$.
        \subref{sfig:sgs-reduction-binary-example}~The gathering instance for the words $10$, $001$, $01$, and $111$.
        A shortest supersequence is $10011$ and the corresponding shortest gathering sequence is $\ldir\udir\,\rdir\udir\,\rdir\udir\,\ldir\udir\,\ldir\udir\,\rdir$.
        The particles gather in the red pixel.
    }\label{fig:sgs-reduction-binary}%
\end{figure}

\begin{theorem}\label{thm:gathering-hard-simple-mazes}
    \quickGathering{} is \NP-hard in simple mazes.
\end{theorem}
\begin{proof}
    Given a set~$S$ of words over $\{0,1\}$, $|S| > 1$, we construct a simple maze as follows:
    We start with a (sufficiently long) row segment~$R$.
    For each word $w \in S$, we create a vertical path gadget consisting of the symbol gadgets shown in \cref{sfig:scs-reduction-binary-gadget-a,sfig:scs-reduction-binary-gadget-b} in the order imposed by~$w$.
    We attach this gadget to the bottom of~$R$, see \cref{sfig:sgs-reduction-binary-example}.
    Clearly, this construction is a maze without holes.
    Also, it is polynomial as it consists of $|S|$ path gadgets, each consisting of at most~$\max_{w \in S} 7|w|$~pixels.
    In each path gadget, we focus on the particle furthest from~$R$ as any gathering sequence for these particles is also a gathering sequence for the~full~polyomino.

    We show that there exists a supersequence for~$S$ of length~$k$ if and only if there exists a gathering sequence of length $2k + 1$ in the corresponding maze.
    Let $\overline{w} \in \{0,1\}^\ast$ be a supersequence for~$S$ with $|\overline{w}| = k$.
    We replace each $0$ in $\overline{w}$ with the moves $\rdir\udir$ and each $1$ with $\ldir\udir$.
    In each of the path gadgets, the particle progresses to the next symbol gadget if and only if the symbol in $\overline{w}$ aligns with the moves for its current symbol gadget.
    Otherwise, it remains in its row.
    Thus, after the $2k$ moves, all particles have progressed to the row segment~$R$.
    To gather them in~$R$'s rightmost pixel, we add one final~$\rdir$ move.

    Conversely, let~$\overline{w}' \in \D^\ast$ be a gathering sequence.
    Without loss of generality, let~$|\overline{w}'|$ be minimal.
    Then, $\overline{w}'$ alternates between horizontal and vertical moves and the first move is horizontal because a vertical move has no effect in the first step.
    Additionally, $\overline{w}'$ does not contain any~$\ddir$-move as particles can only gather in the row segment~$R$ and all shortest sequences to any pixel in~$R$ consist of only the moves $\udir$, $\ldir$, and $\rdir$.
    Finally, since $|S| > 1$, the last move is horizontal as particles from different path gadgets have to gather in the leftmost or rightmost pixel of~$R$.
    Therefore, $\overline{w}'$ is a sequence of length $2k + 1$ for some $k$, consisting of $\rdir\udir$ and $\ldir\udir$ subsequences and ending with a horizontal move.
    We ignore the final move and replace each $\rdir\udir$ with a $0$ and each $\ldir\udir$ with a $1$ to construct a word~$\overline{w}$ of length~$k$.
    As before, a particle progresses to the next symbol gadget exactly if the word in the path gadget contains the corresponding symbol at the position of the particle's gadget.
    Therefore, since~$\overline{w}'$ is a gathering sequence and all particles end up in~$R$, $\overline{w}$ is a supersequence for~$S$.
\end{proof}

Note that the corresponding decision problem is \NP{}-complete as a gathering sequence for a gatherable polyomino can be guessed and verified in polynomial time by \cref{cor:gather-upper-bound}.

The proof of \cref{thm:gathering-hard-simple-mazes} immediately yields \thmW{1}-hardness of \paraGathering.

\begin{corollary}\label{cor:gathering-w1-hard-via-sgs}
\paraGathering is \thmW{1}-hard when parameterized by the size $k$ of the initial configuration, even when the given polyomino $P$ is a gatherable simple maze.
\end{corollary}
\begin{proof}
Consider the reduction in the proof of \cref{thm:gathering-hard-simple-mazes}, but only place particles in the positions furthest from $R$, as depicted in \cref{sfig:sgs-reduction-binary-example}.
This is an fpt-reduction from \textsc{ShortestCommonSupersequence} over a binary alphabet, parameterized by the number of given sequences, to \paraGathering in a gatherable simple maze parameterized by $k$, with matching values for the parameters.
The former has been proven \thmW{1}-hard by Pietrzak~\cite{scs-parameterized}.
\end{proof}

On the other hand, we can approximate the length of a shortest gathering sequence to within a factor of~$4$ in simple mazes.
To achieve this, we first observe that all such instances have exactly one row or column segment where the particles can gather.
Throughout the remainder of this section, we only consider segments consisting of at least two pixels.

\begin{lemma}\label{lem:simple-maze-properties}
    Let~$P$ be a gatherable simple maze with~$N > 1$ pixels.
    Then all row or column segments of~$P$ have at least one end point that is a corner pixel and exactly one segment has two corner pixel end points.
\end{lemma}
\begin{proof}
    Let~$P$ be a gatherable simple maze and let~$R$ be a segment whose end points~$r_1,r_2 \in R$ are not corner pixels.
    By the definition of mazes and segments, $r_1$ and~$r_2$ have degree~$3$ and belong to distinct segments~$S_1$ and~$S_2$ perpendicular to~$R$.
    In particular,~$r_i$ is not an end point of~$S_i$.
    Since~$P$ is gatherable, there must be a pixel in~$V$ reachable from any pixel in both~$S_1$ and~$S_2$.
    Note that $S_1$ and $S_2$ are only connected via~$R$ and their only connection to~$R$ is via~$r_1$ and~$r_2$ as any other path would imply that~$P$ is not simple.
    However, pixels located at the end points of~$S_i$ cannot reach~$r_i$ to enter~$R$,~a~contradiction.

    Now, let~$P$ be a gatherable simple maze and let~$R$ and~$S$ be two distinct segments whose end points~$r_1,r_2 \in R$ and $s_1,s_2 \in S$ are all corner pixels.
    By the definition of mazes, $r_1$, $r_2$, $s_1$, and $s_2$ have degree~$1$.
    Therefore, $\{r_1,r_2\} \cdot v \subseteq \{r_1,r_2\}$ and $\{s_1,s_2\} \cdot v \subseteq \{s_1,s_2\}$ for any $v \in \D$, i.e., no particle located at an end point of~$R$ or~$S$ can leave its segment.
    Consequently, $P$ is not gatherable, a contradiction.

    Finally, assume~$P$ has no segment whose end points are both corner pixels, i.e., all segments have exactly one corner pixel.
    Consider any segment~$R_1$ and its non-corner end point~$r_1$, which has degree~$3$.
    Pixel~$r_1$ is part of another segment~$R_2$ with a different non\nobreakdash-corner end point~$r_2$.
    Extending the argument, any non\nobreakdash-corner end point~$r_i$ of a segment~$R_i$ is part of another segment~$R_{i+1}$ with non\nobreakdash-corner end point~$r_{i+1}$.
    Since~$P$ has only finitely many segments, there must be $i \neq j$ with $R_i = R_j$.
    This implies that $P$ is not simple,~a~contradiction.
\end{proof}

\begin{theorem}\label{thm:simple-maze-approx}
    Let~$P$ be a gatherable simple maze.
    Then we can compute a gathering sequence for $P$ of length at most $4 \cdot \sgs$ in polynomial time.
\end{theorem}
\begin{proof}
    Assume~$N > 1$.
    By \cref{lem:simple-maze-properties}, there exists a unique segment~$R$ whose end points are both corner pixels and any other segment has exactly one corner pixel.
    Clearly, all particles have to gather in~$R$ as particles initially located in~$R$'s corner pixels cannot enter any other segment.
    It suffices to focus on particles located at corner pixels as all other particles are also gathered by gathering sequences for corner pixels in simple mazes.
    Now consider a particle~$p$ located in the corner pixel of a horizontal segment~$S_1 \neq R$.
    To reach~$R$, $p$ has to pass a sequence of alternating horizontal and vertical segments~$S_1, \dots, S_\ell$ with~$S_\ell = R$.
    Therefore, the shortest sequence moving~$p$ to~$R$ consists of alternating $\{\ldir,\rdir\}$ and $\{\udir, \ddir\}$ moves.
    Note that~$p$ cannot return to~$S_i$ after progressing to~$S_{i+1}$ and then moving in any direction perpendicular to~$S_i$ because the intersection pixel between any consecutive segments~$S_i,S_{i+1}$ has degree~$3$ by the definition of mazes.

    Let~$d$ be the maximum length of a shortest sequence from any corner pixel to~$R$ and consider the sequence~$\udir\rdir\ddir\ldir$.
    Any particle initially located at a corner pixel progresses to the next segment with at least one of the four moves.
    Thus, the sequence~$(\udir\rdir\ddir\ldir)^d$ is at most four times as long as the shortest sequence moving all particles to~$R$.
    Finally, the last move in any gathering sequence is perpendicular to~$R$ to gather all particles in one of $R$'s~end~points.
\end{proof}

For arbitrary simple polyominoes, there is likely no polynomial-time algorithm approximating $\sgs$ within any constant factor.
To show this, we extend our reduction from \textsc{ShortestCommonSupersequence} to unbounded alphabets.
If the alphabet size is not fixed, the problem is hard to approximate within any constant factor $\alpha \geq 1$ unless $\P = \NP$~\cite{scs-inapproximability}.

We use the same approach to construct a polyomino as in \cref{thm:gathering-hard-simple-mazes}, i.e., we attach~a path gadget consisting of symbol gadgets for each given word to a baseline segment where particles gather.
It suffices to generalize the symbol gadget for larger alphabets, see \cref{sfig:no-constant-approx-simple-symbol}.

\begin{figure}[tb]
    \begin{subcaptionblock}{326pt}%
        \phantomcaption\label{sfig:no-constant-approx-simple-symbol}%
        \centering%
        \includegraphics[page=1]{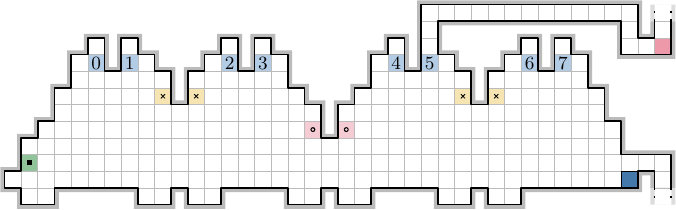}%

        \captiontext*{}%
    \end{subcaptionblock}%
    \hfill%
    \begin{subcaptionblock}{45pt}%
        \phantomcaption\label{sfig:no-constant-approx-simple-start}%
        \centering%
        \includegraphics[page=2]{no-constant-approx-simple}%

        \captiontext*{}%
    \end{subcaptionblock}%
    \caption{%
        \subref{sfig:no-constant-approx-simple-symbol}~The gadget for symbol~$5$ with alphabet $\Sigma = \{0, \dots, 7\}$.
        It takes~$18$ steps to move the bottom particle (or a particle in the entry to its right) to the top red pixel via any of the eight possible exits.
        With~$\udir\ldir$, the particle reaches the green pixel.
        From here, it repeatedly progresses to the next color with one of two possible four-move sequences until it gets to the exit or a dead end.
        If the particle is at the exit, it reaches the top red pixel with~$\udir\rdir\ddir\rdir$.
        Otherwise, this sequence returns the particle to its origin.
        In this example, where the exit is located at symbol~$5$, the full sequence to the top red pixel is $\udir\ldir\,\udir\rdir\udir\ldir\,\ddir\rdir\udir\rdir\,\ddir\ldir\udir\ldir\,\udir\rdir\ddir\rdir$.
        \subref{sfig:no-constant-approx-simple-start}~By attaching this subpath to the start of each path gadget, we ensure that all gathering sequences include the move sequence $\ddir\ldir\udir\ldir\udir\rdir\ddir\rdir$.
        The sequence forces all particles located anywhere in a path gadget to gather at a symbol gadget's entry.
    }\label{fig:no-constant-approx-simple}%
\end{figure}

\begin{theorem}\label{thm:inapproximability-simple}
    Unless $\P = \NP$, there is no polynomial-time approximation algorithm for \quickGathering{} in simple polyominoes with approximation ratio $\alpha \in \bigO(1)$.
\end{theorem}
\begin{proof}
    Without loss of generality, we assume $|\Sigma| = 2^\ell$ for some $\ell \in \mathbb{N}$, adding dummy symbols if necessary.
    Given a set~$S$ of words over $\Sigma$, $|S| > 1$, we use the same idea to construct a polyomino as in \cref{thm:gathering-hard-simple-mazes} by attaching a path gadget consisting of symbol gadgets for each $w \in S$ to a horizontal baseline~$R$, where the particles will gather.
    Our symbol gadget construction is shown in \cref{sfig:no-constant-approx-simple-symbol} for the symbol~$5$ with alphabet $\Sigma = \{0, \dots, 7\}$.
    At the start of each path gadget, we prepend the construction shown in \cref{sfig:no-constant-approx-simple-start}, which ensures that all particles located anywhere in a path gadget in the initial configuration gather at a symbol gadget's entry point in the gadget's bottom-right corner.

    From each symbol gadget's entry point, all possible exits in the top are reachable within~$2 + 4\ell$ moves.
    After an initial~$\udir\ldir$, there are~$\ell$ sequential choices of two different four-move sequences ($\udir\rdir\udir\ldir$ vs.\ $\ddir\rdir\udir\rdir$ or $\udir\ldir\udir\rdir$ vs.\ $\ddir\ldir\udir\ldir$).
    Intuitively, the $i$th choice determines the $i$th bit of the symbol to be emulated.
    If the particle is at an exit, it progresses to the entry of the next symbol gadget with the move sequence $\udir\rdir\ddir\rdir$.
    Otherwise, it returns to the entry of its current symbol gadget with the same sequence.
    A final~$\udir\ldir$ gathers the particles that progressed through all of their symbol gadgets in~$R$.
    Thus, by the same arguments as in \cref{thm:gathering-hard-simple-mazes}, there exists a supersequence of length~$k$ for~$S$ if and only if there exists a gathering sequence of length~$10 + k \cdot (2 + 4\ell)$ in the corresponding simple polyomino.

    By construction, any gathering sequence for the polyomino corresponds to a supersequence for~$S$.
    Therefore, any $c$-approximate solution for \quickGathering{} in simple polyominoes yields a $\Theta(c)$-approximate solution for \textsc{ShortestCommonSupersequence}.
    As \textsc{ShortestCommonSupersequence} is hard to approximate within any constant factor unless $\P = \NP$~\cite{scs-inapproximability}, the same holds for \quickGathering{} in simple polyominoes.
\end{proof}

\section{Gathering subsets in simple polyominoes}\label{sec:subsets}

Like previous reductions for the full tilt model~\cite{tilt-soda20,tilt-algosensors13}, our \PSPACE-hardness proof for \subsetGathering in \cref{cor:subset-gathering-mazes} requires polyominoes with holes.
We can retain a weaker hardness result for simple polyominoes by reducing from \textsc{IntersectionNonEmptiness} of \emph{tally automata}, i.e., automata over the unary alphabet $\{0\}$: Given a family $A_i = (Q_i, \{0\}, \delta_i, q_i, F_i)$ of tally automata, determine whether $\bigcap_{i} L(A_i) \neq \varnothing$, i.e., whether the $A_i$ accept a common word.
The \NP-hardness of this variant goes back to seminal work by Stockmeyer and Meyer~\cite{intersection-unary-stoc73}.
A newer, independent proof by Morawietz, Rehs, and Weller~\cite{intersection-timecop} emphasizes that this variant remains \NP-hard if the graphs underlying the automata are cycles.
Inspection of the proof reveals that we may also assume the lengths of all cycles to be odd and greater than~$1$, since they are chosen to be products of large primes.
Finally, we may assume all automata to have at least one non-accepting state (or the automaton would accept all words).
We implicitly include these restrictions when talking about tally automata in the following.

\begin{figure}[tbh]
  \begin{subcaptionblock}{0.25\textwidth}
    \phantomcaption\label{sfig:tally-gadgets-even-yes}
    \centering
    \includegraphics[page=1]{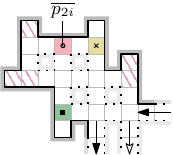}\\
    \captiontext*{}
  \end{subcaptionblock}\hfill
  \begin{subcaptionblock}{0.25\textwidth}
    \phantomcaption\label{sfig:tally-gadgets-odd-yes}
    \centering
    \includegraphics[page=2]{tally-reduction-gadgets}\\
    \captiontext*{}
  \end{subcaptionblock}\hfill
  \begin{subcaptionblock}{0.25\textwidth}
    \phantomcaption\label{sfig:tally-gadgets-even-no}
    \centering
    \includegraphics[page=3]{tally-reduction-gadgets}\\
    \captiontext*{}
  \end{subcaptionblock}\hfill
  \begin{subcaptionblock}{0.25\textwidth}
    \phantomcaption\label{sfig:tally-gadgets-odd-no}
    \centering
    \includegraphics[page=4]{tally-reduction-gadgets}\\
    \captiontext*{}
  \end{subcaptionblock}\medskip

  \begin{subcaptionblock}{0.28\textwidth}
    \phantomcaption\label{sfig:tally-gadgets-start}
    \centering
    \includegraphics[page=5]{tally-reduction-gadgets}\\
    \captiontext*{}
  \end{subcaptionblock}\hfill
  \begin{subcaptionblock}{0.5\textwidth}
    \phantomcaption\label{sfig:tally-gadgets-bridge}
    \centering
    \includegraphics[page=6]{tally-reduction-gadgets}\\
    \captiontext*{}
  \end{subcaptionblock}\hfill
  \begin{subcaptionblock}{0.22\textwidth}
    \phantomcaption\label{sfig:tally-gadgets-end}
    \centering
    \includegraphics[page=7]{tally-reduction-gadgets}\\
    \captiontext*{}
  \end{subcaptionblock}\hfill
  \caption{Gadgets used in the reduction from tally automata intersection. Pixels marked with colors and symbols indicate the cycle $\rdir\ddir\ldir\udir$ between pixels representing states, which is continued between gadgets via the black arrows. The stroked tiling pattern highlights pixels with a safe path to the goal area using the white arrows. Dotted areas specify optional holes that turn the simple polyomino into a maze. \subref{sfig:tally-gadgets-even-yes}~Accepting state with even index. \subref{sfig:tally-gadgets-odd-yes}~Accepting state with odd index. \subref{sfig:tally-gadgets-even-no}~Rejecting state with even index. \subref{sfig:tally-gadgets-odd-no}~Rejecting state with odd index. \subref{sfig:tally-gadgets-start}~Start gadget including $\overline{p_0}$ and access to the goal area to the left. \subref{sfig:tally-gadgets-bridge}~Bridge gadget. \subref{sfig:tally-gadgets-end}~End gadget connecting $\overline{p_0}$ to $\overline{p_1}$.}\label{fig:tally-gadgets}
\end{figure}

Given a family $A_i = (Q_i, \{0\}, \delta_i, q_i, F_i)$, $1 \leq i \leq k$, of $k$ tally automata, we construct a simple polyomino $P(A_1,\ldots,A_k)$ and an initial configuration as follows.

Consider an automaton $A_i$ with $|Q_i| = 2m+1$ for some $m \geq 1$.
Number its states $p_0,\ldots,p_{2m}$ such that $p_0 \notin F_i$ and $p_j \rightarrow p_{(j+1) \bmod |Q_i|}$ for all $0 \leq j < |Q_i|$.
Place a \emph{start gadget} containing a pixel $\overline{p_0}$ representing $p_0$, followed horizontally by $m-1$ \emph{bridge gadgets} and a final \emph{end gadget}; see \cref{sfig:tally-gadgets-start,sfig:tally-gadgets-bridge,sfig:tally-gadgets-end}.
Above the partial construction, $2m$ gadgets are placed, from right to left, for the remaining states $p_j \in \{p_1,\ldots,p_{2m}\}$, depending on two factors: the parity of $j$ and whether $p_j$ is accepting.
Gadgets for the four possible combinations are depicted in \cref{sfig:tally-gadgets-even-yes,sfig:tally-gadgets-odd-yes,sfig:tally-gadgets-even-no,sfig:tally-gadgets-odd-no}, each containing a representative pixel $\overline{p_j}$.

Partial constructions for $A_i$ are stacked vertically and connected to a \emph{goal area} to the left, which is simply a column segment ending in pixels of degree $1$.
For a set $S \subseteq \bigcup_i Q_i$, we define $\overline{S} = \{\overline{q}: q \in S\}$ to be the set of representative pixels for states in $S$.
The initial configuration is $C_0 = \bigcup_i \{\overline{q_i}\}$, i.e., the representative pixels for the initial states.
The basic idea of the reduction is to represent transitions between states by the sequence of moves $\rdir\ddir\ldir\udir$ that relocates particles between representative pixels.
No minimal gathering sequence that starts with $\rdir$ from a configuration $C \subseteq \bigcup_i \overline{Q_i}$, $|C| \geq 2$, will deviate from this prefix, since alternatives lead to dead ends, i.e., pixels of degree $1$.
Starting with $\ldir$ followed by a vertical move, on the other hand, moves all particles from representatives of non-accepting states into a column segment they cannot escape from, whereas particles on representatives of accepting states move to pixels that guarantee a path to the goal area.
\Cref{fig:tally-reduction-example} depicts a small but complete example of the construction for two tally automata.

\begin{figure}[tbh]
  \begin{subcaptionblock}{0.45\textwidth}
    \phantomcaption\label{sfig:tally-automata}
    \centering
    \begin{tikzpicture}[node distance=1cm,on grid,auto,>={Stealth[round]},shorten >=1pt,
                        every state/.style={minimum size=0cm,inner sep=0.05cm}]
      \node[state]                                       (p_6)                {$p_6$};
      \node[state]                                       (p_5) [right=of p_6] {$p_5$};
      \node[state,accepting]                             (p_4) [right=of p_5] {$p_4$};
      \node[state,initial above,initial text=,accepting] (p_3) [right=of p_4] {$p_3$};
      \node[state]                                       (p_2) [right=of p_3] {$p_2$};
      \node[state,accepting]                             (p_1) [right=of p_2] {$p_1$};
      \node[state]                                       (p_0) [below=of p_6] {$p_0$};

      \path[->] (p_1) edge              (p_2)
                (p_2) edge              (p_3)
                (p_3) edge              (p_4)
                (p_4) edge              (p_5)
                (p_5) edge              (p_6)
                (p_6) edge              (p_0)
                (p_0) edge [bend right=20] (p_1);

      \node[state,initial above,initial text=] (q_4) [below=1.5cm of p_0] {$p'_4$};
      \node[state,accepting]                   (q_3) [right=of q_4]       {$p'_3$};
      \node[state,accepting]                   (q_2) [right=of q_3]       {$p'_2$};
      \node[state]                             (q_1) [right=of q_2]       {$p'_1$};
      \node[state]                             (q_0) [below=of q_4]       {$p'_0$};

      \path[->] (q_1) edge              (q_2)
                (q_2) edge              (q_3)
                (q_3) edge              (q_4)
                (q_4) edge              (q_0)
                (q_0) edge [bend right=20] (q_1);
    \end{tikzpicture}\\
    \captiontext*{}
  \end{subcaptionblock}\hfill
  \begin{subcaptionblock}{0.55\textwidth}
    \phantomcaption\label{sfig:tally-polyomino}
    \centering
    \includegraphics{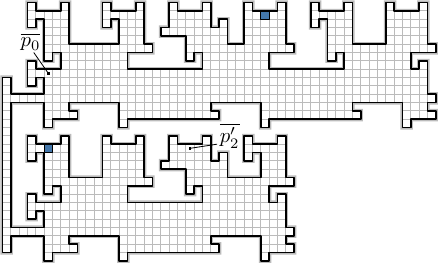}\\
    \captiontext*{}
  \end{subcaptionblock}\hfill
  \caption{Instance of \subsetGathering produced by the reduction from tally automata intersection. \subref{sfig:tally-automata}~Two tally automata, $A_1, A_2$; both accept the word of length 8. \subref{sfig:tally-polyomino}~The simple polyomino $P(A_1,A_2)$; $(\rdir\ddir\ldir\udir)^8\,\ldir\ddir\rdir\ddir\ldir\ddir$ is a gathering sequence for its initial configuration $C_0$.}\label{fig:tally-reduction-example}
\end{figure}

\begin{lemma}\label{lem:tally-accept-packing}
Let $A_i=(Q_i,\{0\},\delta_i,q_i,F_i)$, $1 \leq i \leq k$, be a family of $k$ tally automata and $\ell \geq 0$.
Then, $0^\ell \in \bigcap_i L(A_i)$ if and only if $C_0 \cdot (\rdir\ddir\ldir\udir)^\ell \subseteq \bigcup_i \overline{F_i}$ in $P(A_1,\ldots,A_k)$.
\end{lemma}
\begin{proof}
Observe that, for every $p_j \in Q_i$ and $\ell \geq 0$, $\overline{p_j} \cdot (\rdir\ddir\ldir\udir)^\ell = \overline{p_{(j+\ell) \bmod |Q_i|}}$ holds by construction.
Thus, $q_i \cdot 0^\ell \in F_i$ if and only if $\overline{q_i} \cdot (\rdir\ddir\ldir\udir)^\ell \in \overline{F_i}$.
\end{proof}

The high-level argument of our reduction is that a common word $0^\ell \in \bigcap_i L(A_i)$ leads to a gathering sequence $(\rdir\ddir\ldir\udir)^\ell\,\ldir\ddir\rdir\ddir\ldir\ddir$ for $C_0$ in $P(A_1,\ldots,A_k)$.
For the reverse direction, we argue that if there is any gathering sequence for $C_0$, then it must first move all particles to representatives of accepting states, which can be done with a (shorter) sequence of the form $(\rdir\ddir\ldir\udir)^{\ell'}$, which corresponds to a common word $0^{\ell'}$ for the automata.

\begin{theorem}\label{thm:subset-gathering-simple}
\emph{\subsetGathering} is \NP-hard in simple polyominoes.
\end{theorem}
\begin{proof}
Given a family $A_1,\ldots,A_k$ of tally automata, we construct $P(A_1,\ldots,A_k)$ with initial configuration $C_0$ as described above.
All gadgets have constant size, and we use $\bigO(\sum_i |Q_i|)$ of them, so the reduction runs in polynomial time.
We now prove that the $A_i$ accept a common word if and only if there is a gathering sequence for $C_0$.
The claim then follows from the \NP-hardness of \textsc{IntersectionNonEmptiness} for this type of automata~\cite{complexity-survey,intersection-timecop,intersection-unary-stoc73}.

First, let $0^\ell \in \bigcap_i L(A_i)$.
Then $p \cdot (\rdir\ddir\ldir\udir)^\ell \in \bigcup_i \overline{F_i}$ for every $p \in C_0$, by \cref{lem:tally-accept-packing}.
By construction, $|\bigcup_i \overline{F_i} \cdot \ldir\ddir\rdir\ddir\ldir\ddir| = 1$, as all particles are gathered at the bottom of the goal area.
Thus, $(\rdir\ddir\ldir\udir)^\ell\,\ldir\ddir\rdir\ddir\ldir\ddir$ is a gathering sequence for $C_0$.

Conversely, assume there is a gathering sequence $w \in \D^*$ for $C_0$.
Without loss of generality, assume $w$ to alternate between horizontal and vertical moves.
Pick $i \leq k$, consider the particle originally placed at $\overline{q_i}$, and let $w'$ be the longest prefix of $w$ such that $\overline{q_i} \cdot w' \in \overline{Q_i}$.
Using $\ldir\udir$ or $\ldir\ddir$ when in a state in $\overline{Q_i}$ does not allow a particle to re-enter $\overline{Q_i}$, so this has not happened during $w'$.
On the other hand, this is the only way to break out of the cycle, and all particles move in unison while on their respective cycles.
Thus, by eliminating superfluous detours into dead ends, $w'$ can be turned into an equivalent $w'' \in \{\rdir\ddir\ldir\udir\}^*$, say $(\rdir\ddir\ldir\udir)^{\ell'}$.
Then, $C_0 \cdot w'' \subseteq \bigcup_i \overline{Q_i}$ and the next moves are either $\ldir\udir$ or $\ldir\ddir$.
If $C_0 \cdot w'' \nsubseteq \bigcup_i \overline{F_i}$, then a particle would get trapped, contradicting the assumption that $w$ is a gathering sequence.
Hence, $0^{\ell'} \in \bigcap_i L(A_i)$, by \cref{lem:tally-accept-packing}.
\end{proof}

\begin{corollary}\label{cor:subset-lower-bound-simple}
There are infinitely many simple polyominoes $P$ with gatherable configurations~$C$ such that $\sgs[C] \in 2^{\Omega(\sqrt{n \log n})}$, where $n$ is the number of corners of~$P$.
\end{corollary}
\begin{proof}
We use an idea previously employed by Lee and Yannakakis~\cite{identification-lower-bound} in the context of the verification of finite state machines.
Choose $k$ distinct prime numbers, $\rho_1,\ldots,\rho_k$, and build $k$ tally automata as follows.
For every $1 \leq i \leq k$, let $Q_i=\{p^i_0,\ldots,p^i_{\rho_i-1}\}$, $\delta_i(0)(p^i_j) = p^i_{(j+1) \bmod \rho_i}$ for all $0 \leq j < \rho_i$, and $A_i = (Q_i,\{0\},\delta_i,p^i_0,\{p^i_{\rho_i - 1}\})$.
By the Chinese remainder theorem, the smallest number $\ell$ such that $0^\ell$ is accepted by all $A_i$ is $\ell = \prod_{i=1}^k \rho_i - 1$.
Construct $P(A_1,\ldots,A_k)$ and $C_0$ as before.
Then, $C_0$ is gatherable with $\sgs[C_0] > 4 \ell = 4 \prod_{i=1}^k \rho_i - 1$ (or \cref{lem:tally-accept-packing} would yield a smaller $\ell$).
For every $\lambda \geq 7387$, it is possible to choose a set $S$ of primes such that $\sum_{\rho \in S} \rho \leq \lambda$ and $\prod_{\rho \in S} \rho \geq e^{\sqrt{\lambda \log \lambda}}$, as calculated by Del{\'{e}}glise and Nicolas~\cite{prime-product}.
By choosing such a set (possibly discarding the even prime, $2$) we get $\sgs[C_0] \geq 2 e^{\sqrt{\lambda \log \lambda}}$.
Since $n \in \bigO(\lambda)$, this implies $\sgs[C_0] \in 2^{\Omega(\sqrt{n \log n})}$.
\end{proof}

Note that the optional holes indicated in \cref{fig:tally-gadgets} can turn the family of simple polyominoes into a family of mazes.
Furthermore, \cref{cor:subset-lower-bound-simple} also implies a lower bound of $2^{\Omega(\sqrt{N \log N})}$, because the construction guarantees $n \in \Theta(N)$.
This gives a slightly stronger result for mazes than the $2^{\Omega(\sqrt{N})}$ bound that could have been derived from the construction in \cref{sec:simulation}.

The reduction from intersection of tally automata yields an alternative to \cref{cor:gathering-w1-hard-via-sgs} for proving \paraGathering{} \thmW{1}-hard with parameter $k$ in a different restricted setting.

\begin{corollary}\label{cor:gathering-w1-hard-via-subsets}
\paraGathering is \thmW{1}-hard when parameterized by the size $k$ of the initial configuration, even when the given polyomino $P$ is simple or a maze and the limit on the length of a gathering sequence is $\ell=\infty$.
\end{corollary}
\begin{proof}
Observe that the reduction in \cref{thm:subset-gathering-simple} maps the number of automata to the number of particles $k$, i.e., it is an fpt-reduction from \textsc{IntersectionNonEmptiness} of tally automata parameterized by the number of automata, which is \thmW{1}-hard~\cite{intersection-timecop}.
\end{proof}

\section{Consequences for related problems}\label{sec:consequences}

Our techniques lead to a number of corollaries interesting in their own right, including a novel approach to gathering in \SSt and hardness results for related problems.

\subsection{The single step model}\label{sec:single-step}

We focused on the full tilt model up till now, but our work also has implications for the single step model.
First, we catch up on providing a formal definition of a move to the left in the single step model; other directions are analogous.

\begin{description}
  \item[Single step (blocking):] $p \in \delta_{\SSt}^\times(\ldir)(C)$ if and only if $p$ and all pixels left of $p$ in its row segment are occupied in $C$, or $p + (1,0)^\Tpose \in C$.
  \item[Single step (merging):] $p \in \delta_{\SSt}^\cup(\ldir)(C)$ if and only if $p \in C$ and $p$ is the leftmost pixel in its row segment, or $p + (1,0)^\Tpose \in C$.
\end{description}

Like \FT, \SSt maps singleton configurations to singleton configurations, resulting in a function $\delta_{\SSt}^1: \D \to (V \to V)$.
Gathering sequences in \SSt correspond precisely to synchronizing words in the automaton $(V,\D,\delta_{\SSt}^1)$.
To use this approach, it remains to show how to compute $\delta^1_{\SSt}$ from a representation of a given polyomino $P$.
This is readily accomplished by modifying the dual graph $G_P$: Replace every edge $\{p,q\} \in E$ with two directed, labeled edges $p \xrightarrow{w} q$ and $q \xrightarrow{\overline{w}} p$, where $w, \overline{w} \in \D$ are the directions corresponding to the relative positioning of $p$ and $q$, e.g., $w=\rdir$ and $\overline{w}=\ldir$ if $q = p + (1,0)^\Tpose$. Finally, add self-loops to pixels at the boundary, e.g., $p \xrightarrow{\ldir} p$ if $\{p-(1,0)^\Tpose,p\} \notin E$.
Observing the combinatorial structure of automata $(V,\D,\delta^1_{\SSt})$ and applying a known result from automata theory allows us to immediately derive a new upper bound on the length of gathering sequences in \SSt.

\begin{corollary}\label{cor:single-step-quadratic}
For every polyomino $P$ with $N$ pixels, $\sgs \in \bigO(N^2)$ in \SSt.
\end{corollary}
\begin{proof}
Observe that the labeled digraph constructed from $G_P$, which underlies the automaton, is Eulerian, because every vertex has in-degree and out-degree four. Kari~\cite{sync-eulerian} proved that automata $(Q,\Sigma,\delta)$ with underlying Eulerian digraphs have synchronizing words of length at most $(|Q|-2)(|Q|-1)+1$, if they are synchronizing at all. Since~\cite{gathering-case16}~showed that all polyominoes have gathering sequences in \SSt, the claim follows.
\end{proof}

\Cref{cor:single-step-quadratic} is not particularly exciting in light of the known bounds $\bigO(n_c D^2)$ for general polyominoes and $\bigO(n_c D)$ for simple polyominoes, where $D$ is the diameter of $P$, achieved via geometric arguments in~\cite{gathering-icra20}, although it may lead to shorter gathering sequences in non-simple polyominoes with large diameter $D \in \Theta(N)$.
It is, however, another indication that transferring results from automata theory to tilt problems can yield novel insights.
It would be interesting if the algebraic techniques from~\cite{sync-eulerian} could be combined with the geometric ideas from~\cite{gathering-icra20} for even better results.

\subsection{Established problems in restricted environments}

As mentioned in the introduction, the complexity of \occupancy and \shapeReconfiguration in simple polyominoes is an open problem.
Our approach to \subsetGathering in simple polyominoes implies that both are (at least) \NP-hard.

\begin{corollary}\label{cor:occupancy-hard-simple}
\occupancy and \shapeReconfiguration are \NP-hard in the blocking variant of \FT, even in simple polyominoes or mazes.
\end{corollary}
\begin{proof}
Take the construction from \cref{thm:subset-gathering-simple} for $k$ automata and extend the goal area upwards so that there are $k$ pixels above the topmost intersection.
\Cref{sfig:tally-polyomino} illustrates an example with $k=2$.
The distinction between $\delta^\times$ and $\delta^\cup$ only matters in the goal area, since that is the only segment that can possibly contain more than one particle at a time.
Now the $k$th pixel from the top in the goal area can be occupied~--~and the configuration consisting of the $k$ topmost pixels in the goal are can be reached~--~if and only if the automata accept a common word.
\end{proof}

The \NP-hardness of \occupancy and \shapeReconfiguration in simple polyominoes due to \cref{cor:occupancy-hard-simple} raises the question if the problems get easier when the geometry is restricted further.
Indeed, the problems become (trivially) polynomially solvable in rectangles.

\begin{figure}[htb]
  \centering
  \includegraphics{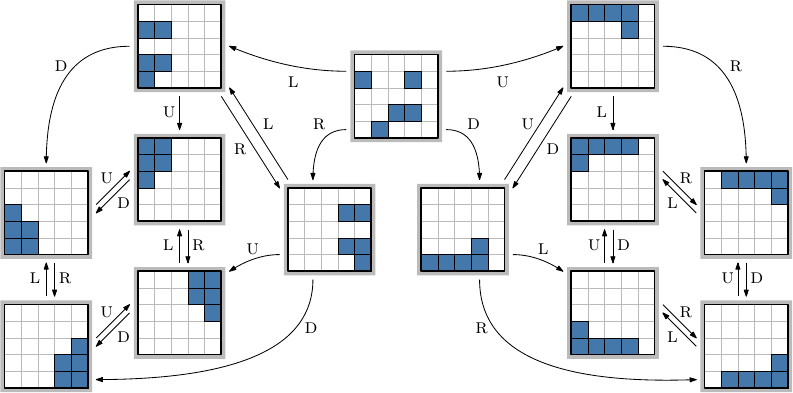}
  \caption{Exhaustive overview of all configurations reachable from an initial configuration inside a rectangle and the moves connecting them. Self-loops are omitted to avoid clutter.}\label{fig:full-tilt-rectangle}
\end{figure}

\begin{observation}\label{obs:occupancy-easy-square}
\occupancy and \shapeReconfiguration are solvable in polynomial time when restricted to polyominoes with rectangular boundaries in \FT.
\end{observation}
\begin{proof}
After two perpendicular moves, a configuration within a rectangle is completely determined by two things: a sequence of numbers describing how many particles are in non-empty columns, and the corner in which the particles congregate.
Only four configurations are reachable from such a state, for a total of at most $13$ configurations that have to be considered, see \cref{fig:full-tilt-rectangle} for an example.
Thus, \occupancy and \shapeReconfiguration in a rectangle can be solved by exhaustive enumeration of a finite number of configurations, each of which can be checked in polynomial time.
\end{proof}

The argument in \cref{obs:occupancy-easy-square} plays a part in the analysis of ``2048 without merging''~\cite{2048-cccg20}, which is essentially the full tilt model in a rectangle with labeled particles.
Labeled configurations exhibit a richer variety of reachable configurations, with ``canonical'' configurations reached after two perpendicular moves corresponding to the elements of a group~of~permutations~\cite{2048-cccg20}.

Like \paraGathering, \occupancy and \shapeReconfiguration are trivially fixed-parameter tractable when parameterized by the maximum length of a sequence of moves.
While it is conceivable that \textsc{SAT}-based reductions for these problems, like the one from~\cite{tilt-algosensors13}, could be adapted to simple polyominoes with some effort, our approach via tally automata has the advantage of immediately providing insights into the parameterized complexity when parameterized by the size of the initial configuration.

\begin{corollary}\label{cor:occupancy-w1-hard}
\occupancy and \shapeReconfiguration are \thmW{1}-hard in the blocking variant of \FT when parameterized by the size $k$ of the initial configuration, even in simple polyominoes or mazes.
\end{corollary}
\begin{proof}
The claim follows directly from the proofs of \cref{cor:occupancy-hard-simple,cor:gathering-w1-hard-via-subsets}.
\end{proof}

\begin{table}[tbh]
  \caption{Overview of complexity results for \shapeReconfiguration and \occupancy.}\label{tbl:complexity-results-classical}
    \centering
    \begin{tabular}{llllll}
      \hline
      & & \multicolumn{2}{l}{\FT} & \multicolumn{2}{l}{\SSt}\\
      \cline{3-6}
      Problem & Geometry & Complexity & Reference & Complexity & Reference\\
      \hline
      \textsc{ShapeReconf.} & General & \PSPACE-c. & \cite[Cor.~5.2]{tilt-soda20} & \PSPACE-c. & \cite[Thm.~2]{tilt-hardness-cccg20}\\
      & Simple & \NP-hard & \textbf{Cor.~\ref{cor:occupancy-hard-simple}} & & \\
      & & \thmW{1}-hard by $k$ & \textbf{Cor.~\ref{cor:occupancy-w1-hard}} & & \\
      & Square & \P & \textbf{Obs.~\ref{obs:occupancy-easy-square}} & \NP-hard & \cite[Thm.~3]{sweeping-eurocg16}\\
      \hline
      \occupancy & General & \PSPACE-c. & \cite[Thm.~5.1]{tilt-soda20} & \P & \cite[Thm.~3.1]{tilt-limited-directions-jip20} \\
      & Simple & \NP-hard & \textbf{Cor.~\ref{cor:occupancy-hard-simple}} & & \\
      & & \thmW{1}-hard by $k$ & \textbf{Cor.~\ref{cor:occupancy-w1-hard}} & & \\
      & Square & \P & \textbf{Obs.~\ref{obs:occupancy-easy-square}} & & \\
      \hline
    \end{tabular}
\end{table}

An overview of new and previously known complexity results for \shapeReconfiguration and \occupancy is given in \cref{tbl:complexity-results-classical}.
The differences between the full tilt and single step models are quite curious.
So far, \occupancy does not look easier than \shapeReconfiguration in \FT, where both exhibit a computational complexity that scales with the geometric complexity of the environment, from polynomially solvable to \PSPACE-complete.
In \SSt on the other hand, \occupancy is easy in all environments (unless you consider more complicated particle shapes~\cite{tilt-hardness-cccg20}), whereas \shapeReconfiguration remains hard even in extremely restricted environments.

\subsection{Hardness of a generalized problem for deterministic tilt cycles}

In the recent paper~\cite{tilt-deterministic}, the authors investigate the complexity of various tilt problems for \emph{deterministic} tilt sequences, i.e., repetitions of a specified sequence $w \in \D^*$.
They show that \occupancy and \shapeReconfiguration remain \PSPACE-complete for deterministic tilt sequences if a single domino is allowed in addition to unit-sized particles and leave the complexity with only $1 \times 1$ particles as an open problem.
While we have not solved this exact problem, our work implies hardness for a natural generalization.

\problem{\tiltCover}{A polyomino $P$, an initial configuration $C$, and a set $S \subseteq V$ of target pixels.}{Decide if there is a sequence $w \in \D^*$ such that $C \cdot w \supseteq S$.}

Note that \occupancy is \tiltCover with $|S| = 1$, whereas \shapeReconfiguration is \tiltCover with $|S| = |C|$.

\begin{corollary}\label{cor:tilt-cover-deterministic}
\tiltCover is \NP-hard in both \FT and \SSt, even when restricted to deterministic tilt sequences in a simple polyomino or maze.
\end{corollary}
\begin{proof}
Recall the reduction from \textsc{IntersectionNonEmptiness} of tally automata presented in \cref{sec:subsets} that constructs a simple polyomino (or maze) $P(A_1,\ldots,A_k)$ and a configuration $C_0$ from given tally automata $A_1,\ldots,A_k$ in polynomial time.
Consider the deterministic tilt cycle $w = \rdir\ddir\ldir\udir$.
All repetitions $w^\ell \in \{\rdir\ddir\ldir\udir\}^*$ keep the particles in separate parts of the polyomino -- they cannot possibly end up in a common segment.
This allows us to use an idea pioneered in~\cite{tilt-hardness-cccg20} and further developed in~\cite{tilt-fill}: Individual bubbles (free pixels) behave identically in \FT and \SSt, and move maximally like particles in \FT, albeit in the opposite direction.
We choose $C = V \setminus C_0$ as initial configuration, $m = \ldir\udir\rdir\ddir$ as deterministic tilt cycle, and $S = V \setminus \bigcup_i \overline{F_i}$ as target pixels (i.e., $S$ contains all pixels except the representatives of accepting states).
Thus, $C \cdot m^\ell \supseteq S$ if and only if $C_0 \cdot w^\ell \subseteq \bigcup_i F_i$, for all $\ell \geq 0$, in both \FT and \SSt.
By \cref{lem:tally-accept-packing}, this means that $S$ can be covered if and only if the $A_i$ accept a common word, which is \NP-hard to decide~\cite{complexity-survey,intersection-timecop,intersection-unary-stoc73}.
\end{proof}

\begin{figure}[tbh]
  \begin{subcaptionblock}{0.5\textwidth}
    \phantomcaption\label{sfig:tilt-cover-particles}
    \centering
    \includegraphics[page=1]{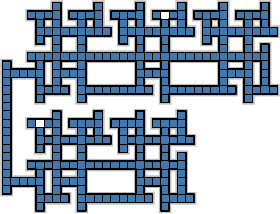}\\
    \captiontext*{}
  \end{subcaptionblock}\hfill
  \begin{subcaptionblock}{0.5\textwidth}
    \phantomcaption\label{sfig:tilt-cover-target}
    \centering
    \includegraphics[page=2]{tilt-cover-example}\\
    \captiontext*{}
  \end{subcaptionblock}\hfill
  \caption{\tiltCover instance derived from the tally automata in \cref{sfig:tally-automata}. \subref{sfig:tilt-cover-particles}~The initial configuration $C$. The free pixels correspond to the initial states. \subref{sfig:tilt-cover-target}~The set $S$ of pixels that must be covered. Pixels corresponding to accepting states are the only ones that need not be covered.}\label{fig:tilt-cover-example}
\end{figure}

In fact, the polyomino can be simplified compared to that in \cref{sec:subsets}, reminiscent of the construction in \cref{thm:gather-lower-bound}; see \cref{fig:tilt-cover-example} for an example maze as an instance of \tiltCover resulting from the tally automata in \cref{sfig:tally-automata}.

\section{Conclusions}\label{sec:conclusion}

We have investigated several problems related to gathering sequences in the full tilt model.
See \cref{tbl:complexity-results-gathering} for an overview our results as well as previous results in the single step model.

The reduction from \cref{sec:subsets} transfers to the related problem of \emph{collecting}~\cite{gathering-case16} a given set of particles into a \emph{connected} configuration in the blocking variant of \FT.
Our hardness results also apply if the particles are to be gathered at a specified position, since one could simply check all possible positions if this was practical.
In the case of \fullGathering, it is still possible to find a gathering sequence $w \in \D^*$ such that $V \xrightarrow{w} \{p\}$ for a given pixel $p$ in polynomial time, or decide that no such sequence exists: Assume there are two distinct pixels $q_1,q_2$ such that $V \xrightarrow{w_i} \{q_i\}$ for $w_1,w_2 \in \D^*$.
Then $\{q_1\} \xrightarrow{w_2} \{q_2\}$ and $\{q_2\} \xrightarrow{w_1} \{q_1\}$, i.e., all pixels where particles can be gathered are mutually reachable.
Thus, one can check if all particles can be gathered at a given pixel $p$ by using the techniques from \cref{sec:gathering} to find a gathering sequence targeting an unspecified pixel $q$, and then check if $p$ is reachable from $q$.
If either step fails gathering at $p$ is impossible.

It may be worth reiterating that most of our algorithmic results are polynomial in the number~$n$ of corners of a given polyomino, whereas our hardness reductions produce the dual graph with $N$ pixels, i.e., the respective problems are hard even for a unary~representation~of~space.

A number of open questions remain.
While closing the gap between the worst-case bounds $\Omega(n^2)$ and $\bigO(n_c n^2)$ for $\sgs$ in general polyominoes looks very hard due to its connection to \Cerny's conjecture, it may be feasible to close the gap in special classes of polyominoes.
Is it easier to prove an upper bound $\sgs \in \bigO(n^2)$ for simple polyominoes?

We have shown \subsetGathering, \occupancy, and \shapeReconfiguration{} to be \NP-hard in simple polyominoes but did not definitively settle their complexity.
Can these problems be shown to remain \PSPACE-complete in simple polyominoes?

Furthermore, we have shown \paraGathering{} with parameter $k$ \thmW{1}-hard for two distinct restricted settings in \cref{cor:gathering-w1-hard-via-sgs,cor:gathering-w1-hard-via-subsets}.
It is likely even harder in its full generality, presumably outside the \W-hierarchy.
Possible starting points for a reduction include a bounded variant of \textsc{IntersectionNonEmptiness} investigated by Wareham~\cite{intersection-wareham}, and work by Fernau and Bruchertseifer~\cite{dfa-sync-parameterized} on the parameterized complexity of synchronizing automata.
For our purposes, however, \W[1]-hardness is strong enough indication that the problem is not fixed-parameter tractable.

Finally, \quickGathering can be approximated within a constant factor in simple mazes, which is \NP-hard in general simple polyominoes.
Is \subsetGathering also easier in simple mazes?
We conjecture it to be polynomially solvable in this case.

\begin{table}[tb]
  \caption{Overview of complexity results for the gathering problems \fullGathering~(FG), \quickGathering~(SGS), \subsetGathering~(SG), and \paraGathering (PG) in both \FT and \SSt, unless $\P = \NP$.}\label{tbl:complexity-results-gathering}
  \centering
  \begin{tabular}{llllll}
    \hline
    & & \multicolumn{2}{l}{\FT} & \multicolumn{2}{l}{\SSt}\\
    \cline{3-6}
    Problem & Geometry & Complexity & Reference & Complexity & Reference\\
    \hline
    FG & General & $\bigO(n_c n^2)$ & \textbf{Thm.~\ref{thm:alg-gathering}} & $\bigO(n_c D^2)$ & \cite[Cor.~10]{gathering-icra20}\\
    & & $\bigO(n^2)$ conj. & \textbf{Conj.~\ref{conj:gathering-upper-bound}} & $\bigO(N^2)$ & \textbf{Cor.~\ref{cor:single-step-quadratic}}\\
    & Simple & $\Omega(n^2)$ & \textbf{Thm.~\ref{thm:gather-lower-bound}} & $\bigO(n_c D)$ & \cite[Cor.~9]{gathering-icra20}\\
    \hline
    SGS & General & No $n^{1 - \varepsilon}$-apx. & \textbf{Cor.~\ref{cor:quick-gathering-inapprox}} & \NP-hard & \cite[Thm.~1]{gathering-icra20}\\
    & Maze & No $N^{\sfrac{1}{2} - \varepsilon}$-apx. & \textbf{Cor.~\ref{cor:quick-gathering-inapprox}}\\
    & Simple & No $\bigO(1)$-apx. & \textbf{Thm.~\ref{thm:inapproximability-simple}}\\
    & Simple maze & \NP-hard, $4$-apx. & \textbf{Thms.~\ref{thm:gathering-hard-simple-mazes},~\ref{thm:simple-maze-approx}}\\
    \hline
    SG & General & \PSPACE-c. & \textbf{Cor.~\ref{cor:subset-gathering-mazes}} (and \cite{tilt-soda20}) & yes & \cite[Sec.~II.A]{gathering-case16}\\
    & Maze & \PSPACE-c. & \textbf{Cor.~\ref{cor:subset-gathering-mazes}}\\
    & Simple & \NP-hard & \textbf{Thm.~\ref{thm:subset-gathering-simple}}\\
    \hline
    PG & General & \FPT by $\ell$ & \textbf{Sec.~\ref{sec:preliminaries}} & \FPT by $\ell$ & trivial\\
    & Simple maze & \thmW{1}-hard by $k$ & \textbf{Cor.~\ref{cor:gathering-w1-hard-via-sgs}}\\
    \hline
  \end{tabular}
\end{table}

\nocite{tilt-limited-directions-ucnc23,assembly-isaac2017,assembly-single-step-wads21,repulsion-swat18,sweeping-fun22,savitch-stoc69,scs-inapproximability-icalp94,complexity-survey-lata09,cerny-original,sync-eulerian-mfcs01,identification-lower-bound-stoc,subset-sync-lower-bound-afl14,inapproximability-binary-csr}
\bibliography{references}

\end{document}